\documentclass[a4paper,11pt, twosided]{article}

\makeatletter
\renewcommand*\@fnsymbol[1]{\the#1}
\makeatother

\usepackage{geometry}
\usepackage{graphicx}
\usepackage{amsmath}
\usepackage{amssymb}
\usepackage{amsthm}
\usepackage{thmtools}
\usepackage{mathtools}
\usepackage{amsfonts}
\usepackage{titlesec}
\usepackage[hidelinks]{hyperref}
\usepackage[capitalize, noabbrev, nameinlink]{cleveref}
\usepackage[british]{babel}
\usepackage{paralist}
\usepackage{enumitem}
\usepackage{subcaption} 
\usepackage{wasysym}
\usepackage{xcolor}
\usepackage[longnamesfirst,round]{natbib}
\usepackage{bm}
\usepackage{bbold}
\usepackage{scalerel}[2016/12/29]  

\usepackage{enumitem}
\usepackage[utf8]{inputenc}

\usepackage[titletoc,title]{appendix}

\usepackage{tikz}
\usepackage{tikz-qtree}
\usetikzlibrary{matrix}
\usetikzlibrary{automata, arrows}

\setlist[itemize]{noitemsep}
\setlist[enumerate]{noitemsep}

\titleformat{\subsection}{\bfseries\large}{\thesubsection}{.5em}{}
\titleformat{\subsubsection}{\normalfont\normalsize\itshape}{\thesubsubsection}{.5em}{}
\titlespacing*{\section}
{0pt}{5ex plus 1ex minus .1ex}{3ex plus .2ex minus .1ex}
\titlespacing*{\subsection}
{0pt}{5ex plus 1ex minus .1ex}{3ex plus .2ex minus .1ex}
\titlespacing*{\subsubsection}
{0pt}{5ex plus 1ex minus .1ex}{3ex plus .2ex minus .1ex}

\newtheoremstyle{THM}{2ex}{1ex}{\itshape}{}{\bfseries}{}{.5em}{\thmname{#1}\thmnumber{ #2}\thmnote{ (#3)}.}
\newtheoremstyle{DEF}{2ex}{1ex}{}{}{\bfseries}{}{.5em}{\thmname{#1}\thmnumber{ #2}\thmnote{ (#3)}.}
\newtheoremstyle{REM}{2ex}{1ex}{}{}{\itshape}{}{.5em}{\thmname{#1}\thmnumber{ #2}\thmnote{ (#3)}.}

\declaretheorem[name=Definition,sharenumber = theorem, style = DEF]{definition}

\declaretheorem[name=Proposition,sharenumber = theorem, style = THM]{proposition}
\declaretheorem[name=Lemma,sharenumber = theorem, style = THM]{lemma}
\declaretheorem[name=Remark,sharenumber = theorem, style = REM]{remark}

\numberwithin{equation}{section}
\numberwithin{figure}{section}

\newcommand{\RR}{\mathbb{R}}

\newcommand{\NN}{\mathbb{N}}
\newcommand{\PP}{\mathbb{P}}
\newcommand{\EE}{\mathbb{E}}
\newcommand{\VV}{\mathbb{V}}
\newcommand{\QQ}{\mathbb{Q}}
\newcommand{\bS}{\mathbb{S}}
\newcommand{\cA}{\mathcal{A}}

\newcommand{\cF}{\mathcal{F}}
\newcommand{\cM}{\mathcal{M}}

\newcommand{\cN}{\mathcal{N}}

\newcommand{\dd}{\mathrm{d}}

\newcommand{\Rint}{R^{\mathrm{int}}_{\bm{S}}}
\newcommand{\Rmon}{R_{\bm{S}}}
\newcommand{\RintX}{R^{\mathrm{int}}_{\bm{S}}(\bm{X})}
\newcommand{\RmonX}{R_{\bm{S}}(\bm{X}_T)}
\newcommand{\RintXhat}{\hat{R}^{\mathrm{int}}_{\bm{S}}(\bm{X})}

\newcommand{\Xlambda}{\bm{X}^{\bm{\lambda},\bm{S}}_T}
\newcommand{\Xlambdan}{\bm{X}^{\bm{\lambda}_n,\bm{S}}_T}

\newcommand{\XT}{\bm{X}_T}
\newcommand{\ST}{\bm{S}_T}
\newcommand{\xo}{\bm{x}_0}
\newcommand{\so}{\bm{s}_0}

\DeclareMathOperator*{\argmin}{arg\,min}

\DeclareMathOperator*{\VaR}{\mathrm{VaR}}

\DeclareMathOperator*{\ES}{\mathrm{ES}}

\DeclareMathOperator*{\conv}{\mathrm{conv}}

\def\be{\begin{equation} \label}
\def\ee{\end{equation}}

\newcommand{\htimes}{\,\scaleobj{0.8}{\odot}\,}
\newcommand{\hdiv}{\,\scaleobj{0.8}{\oslash}\,}

\definecolor{myGreen}{RGB}{39,120,7} 
\definecolor{myBlue}{RGB}{31,88,169}
\definecolor{myGrey}{RGB}{232,237,246} 

\newcommand*\samethanks[1][\value{footnote}]{\footnotemark[#1]}

\title{Set-valued intrinsic measures of systemic risk}
\author{Jana Hlavinov{\'a}\thanks{Vienna University of Economics and Business, Institute for Statistics and Mathematics, Welthandelsplatz 1, 1020 Vienna, Austria, \href{mailto:jana.hlavinova@wu.ac.at}{jana.hlavinova@wu.ac.at} and \href{mailto:birgit.rudloff@wu.ac.at}{birgit.rudloff@wu.ac.at}} \hspace{1cm} Birgit Rudloff\samethanks[1] \hspace{1cm} Alexander Smirnow\thanks{University of Zurich, Department of Banking and Finance, Plattenstrasse 14, 8032 Zurich, Switzerland, \href{mailto:alexander.smirnow@bf.uzh.ch}{alexander.smirnow@bf.uzh.ch}}}

\begin{document}
\maketitle

\begin{abstract}
\noindent In recent years, it has become apparent that an isolated microprudential approach to capital
adequacy requirements of individual institutions is insufficient. 
It can increase the homogeneity of the financial system and ultimately the cost to society. 
For this reason, the focus of the financial and mathematical literature has shifted towards the macroprudential regulation of the financial network as a whole. 
In particular, systemic risk measures have been discussed as a risk measurement and mitigation tool. 
In this spirit, we adopt a general approach of multivariate, set-valued risk measures and combine it with the notion of intrinsic risk measures. 
In order to define the risk of a financial position, intrinsic risk measures utilise only internal capital, which is received when part of the currently held assets are sold, instead of relying on external capital.
We translate this methodology into the systemic framework and show that systemic intrinsic risk measures have desirable properties such as the set-valued equivalents of monotonicity and quasi-convexity. 
Furthermore, for convex acceptance sets we derive a dual representation of the systemic intrinsic risk measure.
We apply our methodology to a modified Eisenberg-Noe network of banks and discuss the appeal of this approach from a regulatory perspective, as it does not elevate the financial system with external capital.
We show evidence that this approach allows to mitigate systemic risk by moving the network towards more stable assets.

\bigskip	
	
\end{abstract}

\newpage

\section{Introduction}
A key task of risk management is the quantification and assessment of the riskiness of a certain financial position, a portfolio, or even a financial system.
\citet{bib:ADHE} were the first to axiomatically define risk measures with a specific management action in mind, namely, raising external capital and holding it in a reference asset.
We call these measures monetary risk measures, a term coined by \citet{bib:FSconvex}, and interpret their value as the minimal monetary amount to be added either directly or through a reference asset to the existing financial position to make it acceptable.
\citet{bib:WPC1, bib:WPC2} formalised this approach for general eligible assets, including the case of defaultable investment vehicles.
Monetary risk measures have a straightforward operational interpretation and, seemingly, they can be directly applied as a risk mitigation tool in the real world.
However, the question of how and at which cost it is possible to raise the necessary capital remains unanswered.
External capital is often not readily available and therefore \citet{bib:FSint} explore the methodology of using internal resources only. 
They suggest a different management action to alter the existing position, namely, selling a fraction of the current unacceptable position and investing the acquired funds into a general eligible asset as defined in \citep{bib:WPC1, bib:WPC2}.
This way, in addition to gaining the benefit of holding capital in a safe eligible asset, the existing unacceptable position is reduced.
Choosing a suitable eligible asset, for example one with negative correlation to the existing position, can further reduce the risk of the new altered position.
In \citep{LAUDAGE2022254}, the concept of intrinsic risk measures of \citet{bib:FSint} is combined with scalar multi-asset risk measures of \citet{bib:WPC3} to consider the multi-asset intrinsic risk of a random variable.

\medskip

The question of capital adequacy is of even higher importance in the setting of financial systems. 
During the last two decades and especially considering the events during the Great Recession from 2007 to 2009, it has become apparent that applying a scalar risk measure to the financial positions of each participant in a system individually, and thereby ignoring dependencies within the system, is not an appropriate approach to measure systemic risk.
Since then, many contributions have improved the understanding of systemic risk and discussed necessary countermeasures.
A discussion of the events during the recession and the economic mechanisms behind them is given in \citep{bib:brunnermeier09}.
To better quantify systemic risk, \citet{bib:adrian16} propose a conditional version of Value-at-Risk (CoVaR) and \citet{Acharya2016} introduce the Systemic Expected Shortfall (SES), both to measure the contribution of each financial entity to the overall risk in the system.
An extensive survey of systemic risk measures including publications before 2016 is provided in \citep{bib:systemic_risk_survey}.
\citet{ChenIyengarMoallemi2013} go on to employ an axiomatic approach to define measures of systemic risk. 
Their approach results in risk measures of the form $\rho(\Lambda(\bm{Y}))$, where $\rho$ is a scalar risk measure, $\Lambda\colon\RR^d\to\RR$ is a non-decreasing aggregation function and $\bm{Y}$ is a $d$-dimensional random vector representing wealths or net worths of each player in a financial system.
\citet{bib:FRW} then argue that applying a scalar risk measure to the aggregated outcome of the system leads to identifying the bailout costs.
These, however, are the costs of saving the system after it has been disrupted, rather than capital requirements that would prevent the system from experiencing distress.
Moreover, using just a single number to quantify the systemic risk of a system with $d\geq2$ participants, important information can get lost, for instance the way in which different participants contribute to the overall risk.

Motivated by this realisation, \citet{bib:FRW} introduce set-valued measures of systemic risk. 
In their framework, the risk measure of a financial system is a collection of all vectors of capital allocations that, added to the individual participants' positions, yield an acceptable system. 
The system is deemed acceptable if the random variable describing the aggregated outcome of the system is an element of the acceptance set.
Choosing an appropriate vector from the set-valued risk measure, capital requirements can be posed by the regulator.
The approach of \citet{bib:FRW} is general in the sense that many risk measures such as the ones in \citep{ChenIyengarMoallemi2013} and \citep{bib:adrian16} can be embedded into their framework.

\medskip

In this article, we combine the management action of intrinsic risk measures introduced in \citep{bib:FSint} and the set-valued approach to measure systemic risk introduced in \citep{bib:FRW}.
For a financial system with $d\geq2$ participants we define the set-valued intrinsic measure of systemic risk as the collection of all vectors of fractions $\bm{\lambda} \in [0,1]^d$ such that if the participants sell the respective fraction of their assets and invest this raised capital in an eligible asset, the aggregated system will be acceptable.
In the framework of a simulated network, we show evidence that with appropriate aggregation functions these acceptable aggregated systems are less volatile, have milder worst case outcomes, and are more likely to repay more of their liabilities to society, compared to the monetary approach.
Intrinsic systemic risk measures therefore provide not only an alternative to measuring risk as capital injections, but also provide alternative risk-reducing management actions that are practical in cases where external monetary injections are unfavourable. 

\medskip

The rest of the paper is structured as follows. 
In \cref{sec:prelim}, we introduce the terminology and lay down the mathematical framework.
We briefly discuss the underlying notion of acceptability and recapitulate scalar risk measures and set-valued measures of systemic risk. 
In \cref{sec:intrinsic_systemic_risk}, we introduce our novel set-valued intrinsic risk measures.
We derive their properties and juxtapose them to scalar intrinsic risk measures and set-valued measures of systemic risk.
Furthermore, we present two algorithms to approximate intrinsic systemic risk measurements numerically. 
We derive a dual representation of intrinsic systemic risk measures in \cref{sec:dual_representation}.
Finally, we apply set-valued intrinsic risk measures to an Eisenberg-Noe network including a sink node and highlight how they can provide asset allocations which improve an unacceptable financial system in \cref{sec:network}.
In \cref{sec:conclusion}, we conclude our findings and discuss possible extensions and further research avenues.

\section{Monetary, intrinsic, and set-valued risk measures} \label{sec:prelim}

In this section, we review important terminology from risk measure theory.
We define acceptance sets axiomatically and use them to define monetary and intrinsic risk measures.
Then we proceed to introduce set-valued measures of systemic risk with general eligible assets.

\medskip

We briefly introduce our notation.
Throughout this paper we work on a probability space $(\Omega, \cF, \PP)$.
We employ a one-period model from time $t=0$ to $t=T$.
Financial positions or assets have a known initial value at time $t=0$ and a random value at a fixed point in time $t = T$. 
These scalar future outcomes are represented by random variables in $L^p = L^p(\Omega,\cF,\PP;\RR)$, $p\in[1,\infty]$, the space of equivalence classes of $p$-integrable random variables endowed with the $L^p$-norm or, in the case $p=\infty$, of essentially bounded random variables endowed with the 
weak$^*$ topology $\sigma(L^\infty,L^1)$.
Future outcomes of $d\in\mathbb{N}$, $d\geq2$ parties in a financial network are represented by multivariate random variables in $L^p_d = L^p(\Omega, \cF, \PP; \RR^d)$ endowed with the canonical norm induced by the $L^p$-norm (respectively the weak$^*$ topology if $p=\infty$) and the $p$-norm on $\RR^d$.

We indicate scalar future outcomes by capital letters with the subscript $T$, for example $X_T$, and their known initial prices by lower-case letters with the subscript $0$, for example $x_0$.
In the multivariate case, we use a bold font to simplify the differentiation, for example $\bm{X}_T$ and $\bm{x}_0$.

We use the componentwise ordering $\leq$ on $\RR^d$, that means for $\bm{x}, \bm{y} \in \RR^d$ we write $\bm{x} \leq \bm{y}$ if and only if $x_k \leq y_k$ for all $k \in \{1,\ldots,d\}$, and we write $\bm{x} < \bm{y}$ if all inequalities are strict.
Furthermore, we use the notation $\RR_+^d=\{\bm{x}\in\RR^d \mid \bm{x}\geq0\}$ and $\RR_{++}^d=\{\bm{x}\in\RR^d \mid \bm{x} > 0\}$ to denote the non-negative and positive orthant of $\RR^d$, respectively.
For $d=1$ we suppress the superscript.
For any set $A \subseteq \RR^d$ we denote its power set by $\mathcal{P}(A)$.

We define the cones $(L^p_d)_+ = \{ \bm{X}_T \in L^p_d \mid  \bm{X}_T \geq 0 \; \PP\text{-a.s.}\}$ and $(L^p_d)_{++} = \{ \bm{X}_T \in L^p_d \mid  \bm{X}_T > 0 \; \PP\text{-a.s.} \}$.

Finally, we denote the Hadamard product and the Hadamard division of $\bm{x},\bm{y}\in\RR^d$ by  
\begin{align*}
\bm{x} \htimes \bm{y} = (x_1 y_1, \ldots, x_d y_d)^\intercal \in \RR^d 
\quad  \text{and} \quad 
\bm{x} \hdiv \bm{y} = \left(\frac{x_1}{y_1}, \ldots, \frac{x_d}{y_d}\right)^\intercal \in \RR^d \,,
\end{align*}
respectively.

\medskip

\subsection{Monetary and intrinsic risk measures}

In risk measure theory, we can differentiate between acceptable and non-acceptable outcomes with the help of acceptance sets.

\begin{definition}\label{defn:acceptance}
A set $\cA\subset L^p$ is called an \emph{acceptance set} if it satisfies
\begin{enumerate}[label = (A\arabic*)]
\item \emph{Non-triviality}: $\cA$ is neither empty nor the whole $L^p$ space, and \label{propA:nontrivial}
\item \emph{Monotonicity}: if $X_T \in \cA$ and $Y_T \in  L^p$ with $X_T \leq Y_T$ $\PP$-a.s., then $Y_T \in \cA$. \label{propA:monotone}
\end{enumerate}
We call the future outcome $X_T$ of a financial position \emph{acceptable} if and only if $X_T\in \cA$.
We will often assume that $\cA$ is a \emph{closed} set, that is $\cA$ is equal to its closure, $\cA = \bar{\cA}$. 
Furthermore, we say an acceptance set $\cA$ is a \emph{cone} if it satisfies
\begin{enumerate}[label = (A\arabic*)]
\setcounter{enumi}{2}
	\item \emph{Conicity}: for all $c > 0$ and $X_T \in \cA$ we have $c X_T \in \cA$,\label{propA:cone}
\end{enumerate} 
and we call $\cA$ \emph{convex} if it satisfies 
\begin{enumerate}[label = (A\arabic*)]
\setcounter{enumi}{3}
	\item \emph{Convexity}: for all $\alpha \in[0,1]$ and $X_T, Y_T \in \cA$ we have $\alpha X_T + (1-\alpha) Y_T \in \cA$.
\label{propA:convex}
\end{enumerate} 
\end{definition}

Properties \ref{propA:nontrivial} and \ref{propA:monotone} are intuitively desirable. 
The acceptance set should be non-trivial such that there are both acceptable and non-acceptable outcomes. 
Monotonicity implies that any position with an outcome that is $\PP$-a.s.~greater than or equal to the outcome of an acceptable position is also acceptable.
Conicity implies that any acceptable position can be scaled by a positive factor and still be acceptable. 
Convexity is of importance when discussing diversification, since it implies that any convex combination of any two acceptable positions is also acceptable.

\bigskip

To quantify the risk of a financial position, we can use this binary structure of acceptability and non-acceptability imposed on the underlying space to construct risk measures.

To this end, we first introduce an \emph{eligible asset} as a tuple $S=(s_0,S_T)\in\RR_{++}\times (L^p)_+$.
Such an eligible asset is a traded asset with initial unitary price $s_0$ and random payoff $S_T$ at time $T$ and serves as an investment vehicle.
For more information on this form of eligible assets see \citep{bib:WPC1,bib:WPC2}.
 
Monetary risk measures quantify the risk of a random variable by its distance to the boundary of the acceptance set.
The distance is measured by the additional monetary amount that needs to be added through the eligible asset to the current financial position to make it acceptable. 
\begin{definition} \label{def:monetary_rm}
    Let $S \in \RR_{++}\times (L^p)_+$ be an eligible asset and let $\cA \subset L^p$ be an acceptance set.
    A \emph{monetary risk measure} ${\rho_{\cA,S} \colon  L^p \to \RR \cup \{+\infty , -\infty\}}$ is defined by
    \begin{align*}
\rho_{\cA,S}(X_T)=\inf\Big\{ m\in\RR \mid X_T+\frac{m}{s_0}S_T\in\cA \Big\}.
\end{align*}
\end{definition}

By definition and by the structure of the acceptance set, monetary risk measures directly satisfy the following properties:

\begin{enumerate}[label=(R\arabic*)]
\item \label{propR:non_constant} \emph{Nonconstant}: $\rho_{\cA,S}$ has at least two distinct values, 
\item \label{propR:cashadditivity}\emph{$S$-additivity}: for all $X_T\in L^p$, $k\in\RR$ we have $\rho_{\cA,S}\big(X_T+k\frac{S_T}{s_0}\big)=\rho_{\cA,S}(X_T)-k$,
\item \label{propR:monotonicity} \emph{Monotonicity}: for all $X_T , Y_T \in  L^p$ if $X_T \leq Y_T$ $\PP$-a.s., then $\rho_{\cA,S}(X_T) \geq \rho_{\cA,S}(Y_T)$.
\end{enumerate}
Properties \ref{propR:non_constant}-\ref{propR:monotonicity} constitute the basic structure of monetary risk measures. 
They are direct consequences of the definitions of $\cA$ and $\rho_{\cA,S}$ and are shown in the proof of Proposition 3.2.3 in \citep{bib:Munari}.

Furthermore, we call a monetary risk measure \emph{positively homogeneous}, if it satisfies 
\begin{enumerate}[label=(R\arabic*)]
\setcounter{enumi}{3}
\item \label{propR:posHomog} \emph{Positive homogeneity}: for all $c>0$, $X_T\in L^p$ we have $\rho_{\cA,S}(cX_T) = c\rho_{\cA,S}(X_T)$.
\end{enumerate}

Finally, we call a monetary risk measure \emph{coherent}, as defined by \citet{bib:ADHE}, if in addition to properties \ref{propR:non_constant} through \ref{propR:posHomog} it satisfies
\begin{enumerate}[label=(R\arabic*)]
\setcounter{enumi}{4}
\item \label{propR:subadditivity} \emph{Subadditivity}: for all $X_T,Y_T\in L^p$ we have $\rho_{\cA,S}(X_T+Y_T)\leq\rho_{\cA,S}(X_T)+\rho_{\cA,S}(Y_T)$.
\end{enumerate}

It is noteworthy that under positive homogeneity, subadditivity is equivalent to
\begin{enumerate}[label=(R\arabic*)]
\setcounter{enumi}{5}
\item \label{propR:convexity} \emph{Convexity}: for all $\alpha \in [0 , 1]$, and all $X_T,Y_T \in  L^p$ we have
\begin{align*}
\rho_{\cA,S}(\alpha X_T+(1-\alpha)Y_T)\leq\alpha\rho_{\cA,S}(X_T)+(1-\alpha)\rho_{\cA,S}(Y_T) \,.
\end{align*}	
\end{enumerate}

Therefore, convex risk measures as in \citep[Definition 4.4]{bib:FSconvex} naturally arise from coherent risk measures by dropping property \ref{propR:posHomog} and substituting \ref{propR:subadditivity} with \ref{propR:convexity}.

\medskip

Acceptance sets and monetary risk measures exhibit a canonical correspondence which allows us to define one from the other.
In \cref{def:monetary_rm}, we used an acceptance set to define a monetary risk measure.
\cref{prop:accset_rm_correspondence} describes how a functional satisfying properties \ref{propR:non_constant}-\ref{propR:monotonicity} defines an acceptance set.
Furthermore, it shows how properties \ref{propR:posHomog} and \ref{propR:convexity} can directly be inferred from \ref{propA:cone} and \ref{propA:convex} and vice versa.

\begin{proposition} \label{prop:accset_rm_correspondence}
Let $\cA$ be an acceptance set and let $\rho_{\cA,S}$ be the corresponding monetary risk measure.
If $\cA$ is a cone, then $\rho_{\cA,S}$ is positively homogeneous, and if $\cA$ is convex, then $\rho_{\cA,S}$ is convex.\\

On the other hand, any functional $\rho_S \colon  L^p \to \RR$ satisfying properties \ref{propR:non_constant}-\ref{propR:monotonicity} defines an acceptance set $\cA_\rho$ via
\begin{align*}
\cA_\rho=\left\{X_T\in L^p \mid \rho_S(X_T)\leq0\right\} .
\end{align*}
If $\rho_S$ is positively homogeneous, then $\cA_\rho$ is a cone, and if $\rho_S$ is convex, then $\cA_\rho$ is convex.

In particular, $\rho_{\cA_\rho,S} = \rho_S$ and $\cA_{\rho_{\cA,S}}\subseteq\cA$ with equality if $\cA$ is closed.
\end{proposition}
\begin{proof}
See the proofs of propositions 4.6 and 4.7 in \citep{bib:FS} for the case of bounded random variables, or the proofs of propositions 3.2.3, 3.2.4, 3.2.5, and 3.2.8 in \citep{bib:Munari} for any real ordered topological vector space.
\end{proof}

The approach with general eligible assets is versatile and, in particular, allows the use of defaultable investment vehicles. 
Cash-additive risk measures as in \citep{bib:ADHE} are a special case with $\frac{S_T}{s_0} \equiv r \in \RR_{++}$.
Two monetary risk measures we will use to construct acceptance sets via \cref{prop:accset_rm_correspondence} in \cref{sec:network} are Value-at-Risk and Expected Shortfall.
\begin{definition}
Let $X \in L^p$ and $\alpha \in (0,1)$. 
The Value-at-Risk at level $\alpha$ is defined as
\begin{align*}
\VaR\nolimits_\alpha (X) = \inf \{ m \in \RR \mid \PP[X + m < 0] \leq \alpha\} \,,
\end{align*}
and the Expected Shortfall at level $\alpha$ is defined as
\begin{align*}
\ES\nolimits_\alpha (X) = \frac{1}{\alpha} \int_0^\alpha \VaR\nolimits_\beta(X) \, d\beta \,.
\end{align*}
\end{definition}

Intrinsic risk measures are also defined via eligible assets and acceptance sets, but their underlying management operation is different.
The methodology of \citet{bib:FSint} relies on restructuring the unacceptable position instead of adding external capital.
\begin{definition}
Let $\cA$ be an acceptance set and let $S=(s_0,S_T)\in\RR_{++}\times \cA$ be an eligible asset. 
The \emph{intrinsic risk measure} is a functional $\rho^\mathrm{int}_{\cA, S} \colon \RR_{++}\times L^p \to [0,1]$ defined as 
\begin{equation} \label{intRS}
\rho^\mathrm{int}_{\cA, S}(X) = \inf\Big\{\lambda\in\left[0,1\right] \mid (1-\lambda)X_T+\lambda\frac{x_0}{s_0}S_T\in\cA\Big\}.
\end{equation}
\end{definition}

In this approach, a percentage $\lambda \in [0,1]$ of the initial position $X_T$ is sold at the initial price $x_0$.
The received monetary value $\lambda x_0$ is then reinvested in the eligible asset with return $\frac{S_T}{s_0}$ resulting in the convex combination
\begin{align} \label{eq:intrinsic_position}
    X_T^{\lambda,S} \coloneqq (1-\lambda)X_T + \lambda \frac{x_0}{s_0} S_T \,.
\end{align}
The intrinsic risk measure returns the smallest $\lambda$ such that $X_T^{\lambda,S}$ is acceptable, at least in the extended sense that $X_T^{\lambda,S} \in \bar{\cA}$, for suitable choices of $S_T$ and $\cA$.
For more details see the discussion below Definition 3.1 in \citep{bib:FSint}.

\medskip

Similarly to monetary risk measures, intrinsic risk measures are decreasing when we choose a suitable order on $\RR_{++} \times  L^p$.
However, in contrast to the convex correspondence between an acceptance set and its monetary risk measure, convexity of the acceptance set corresponds to quasi-convexity of the intrinsic risk measure.
The proof of the following assertions can be found in \citep[Proposition 3.1]{bib:FSint}.

Let $\cA$ be an acceptance set including $0$ and let $S \in \RR_{++} \times \cA$ be an eligible asset. 
Then the intrinsic risk measure satisfies
\begin{enumerate}
    \item \emph{Element-wise monotonicity}: if $x_0 \geq y_0$ and $X_T \geq Y_T$, then $\rho^\mathrm{int}_{\cA, S}(X) \leq \rho^\mathrm{int}_{\cA, S}(Y)$.
    \item \emph{Return-wise monotonicity}: if $\cA$ is a cone and $\frac{X_T}{x_0} \geq \frac{Y_T}{y_0}$, then $\rho^\mathrm{int}_{\cA, S}(X) \leq \rho^\mathrm{int}_{\cA, S}(Y)$.
    \item \emph{Quasi-convexity}: if $\cA$ is convex, then for all $\alpha \in [0,1]$ and all ${X,Y \in \RR_{++} \times  L^p}$ we have
    \begin{align*}
        \rho^\mathrm{int}_{\cA, S}(\alpha X + (1-\alpha)Y) \leq \max\{ \rho^\mathrm{int}_{\cA, S}(X) , \rho^\mathrm{int}_{\cA, S}(Y) \} \,.
    \end{align*}
\end{enumerate}

Quasi-convexity precisely represents the diversification principle, which states that diversification should not increase risk.
Intrinsic risk measures with respect to convex acceptance sets satisfy this principle and do not penalise diversification.
On the other hand, monetary risk measures use convexity to represent the diversification principle, which is possible, since convexity and quasi-convexity are equivalent under $S$-additivity,
see also \citep{bib:CMMM} and the references therein for the case with $s_0=1$ and $S_T=1$ $\PP$-a.s.~ (i.e.\ cash-additive).

We show in \cref{claim:systemic_intrinsic_el_monotonicity} and \cref{claim:systemic_intrinsic_quasi-convex} that these properties translate to the intrinsic measure of systemic risk.

\subsection{Measures of systemic risk}
\label{sec:measures_systemic_risk}
While monetary and intrinsic risk measures can quantify the risk of a single isolated financial institution or position, they are not suitable to measure the systemic risk of a financial system. 
In this section, we turn to financial systems with $d \geq 2$ participants.
As mentioned in the introduction, risk measures of the form $\rho \circ \Lambda: L^p_d \to \RR$ might discard crucial information, as capital is added after aggregation to the whole system so that identifying and understanding the source of the risk becomes difficult.
In order to be most useful for regulators to recognise and mitigate systemic risk and prevent cascades of risk, one can add risk capital to each institution separately before aggregation.
This way, the risk capital of single participants or groups of participants can be adjusted while observing the effects on the whole system.
For this reason we adopt the set-valued approach of \citet{bib:FRW}, where the authors search for all capital allocations $\bm{k} \in \RR^d$ such that $\Lambda(\bm{X}_T + \bm{k})$ belongs to the acceptance set $\cA$.
For the sake of consistency, we generalise Definition 2.2 in \citep{bib:FRW} to $\bm{S}$-additive systemic risk measures. 

\medskip

Consider an interconnected network of $d \geq 2$ financial institutions enumerated by $\{1, \ldots, d \}$.
Let the random vector $\bm{X}_T = (X_T^1 , \ldots , X_T^d)^\intercal \in L^p_d$ denote their future wealths.
In order to use univariate acceptance sets, we need the concept of an aggregation function as a mechanism to map random vectors to univariate random variables. 

\begin{definition}
An \emph{aggregation function} is a non-constant, non-decreasing function $\Lambda\colon\RR^d\to\RR$.
This means that $\Lambda$ has at least two distinct values and that if $\bm{x}\leq \bm{y}$, then $\Lambda(\bm{x})\leq\Lambda(\bm{y})$.
\end{definition}

Furthermore, an aggregation function can be concave or positively homogeneous. 
If these properties are required, we list them explicitly.
For an overview and a discussion of specific examples of aggregation functions see \cite[Example 2.1]{bib:FRW} and the references given therein.

\begin{definition} \label{def:setval_systemic_risk_measure}
Let $\cA$ be an acceptance set, let $\Lambda$ be an aggregation function, and let $\bm{S}=(\bm{s}_0,\bm{S}_T)\in\RR^d_{++}\times (L^p_d)_+$ be a vector of eligible assets. 
A \emph{set-valued measure of systemic risk} is a functional $R_{\bm{S}}\colon L^p_d \to \mathcal{P}(\RR^d)$ defined by 
\begin{align} \label{eq:systRS}
R_{\bm{S}}(\bm{X}_T) &=\big\{\bm{k}\in\RR^d \mid \Lambda(\bm{X}_T + \bm{k} \htimes \bm{S}_T\hdiv \bm{s}_0) \in \cA \big\} \,.
\end{align}
\end{definition}

We will refer to these measures also as \emph{monetary} measures of systemic risk, as an additional monetary amount $\bm{k}$ is added to the system.
This will make it easier to differentiate these measures from the intrinsic type defined in \cref{def:intr_meas_sys_risk}.

\medskip

In the following, we collect some important properties of $R_{\bm{S}}$.
The proofs can be found in \cref{appendix_proofs}. 
Note that these proofs generalise the corresponding proofs of the properties for the case $\bm{s}_0=\bm{1}$ and $\bm{S}_T=\bm{1}$ $\PP$-a.s.~which can be found in \cite{bib:AR} and \cite{bib:FRW}.
\begin{proposition} \label{prop:properties_syst_risk}
Let $\cA$, $\Lambda$ and $\bm{S}$ be as in \cref{def:setval_systemic_risk_measure}.
Then $R_{\bm{S}}\colon L^p_d \to \mathcal{P}(\RR^d)$ as defined in \eqref{eq:systRS} satisfies the following properties:
\begin{enumerate}[label=(\roman*)]
    \item \label{item_uppersets} \emph{Values of $R_{\bm{S}}$ are upper sets:} for all $\bm{X}_T\in L^p_d\colon R_{\bm{S}}(\bm{X}_T)=R_{\bm{S}}(\bm{X}_T)+\RR^d_+$.
    \item \label{item_S-add} $\bm{S}$\emph{-additivity:} for all $\bm{X}_T\in L^p_d$ and $\bm{\ell}\in\RR^d\colon R_{\bm{S}}(\bm{X}_T+\bm{\ell}\htimes\bm{S}_T\hdiv\bm{s}_0)=R_{\bm{S}}(\bm{X}_T)-\bm{\ell}$.
    \item \label{item_monoton} \emph{Monotonicity:} for any $\bm{X}_T, \bm{Y}_T \in  L^p_d$ with $\bm{X}_T \leq \bm{Y}_T \; \PP\text{-a.s.}\colon R_{\bm{S}}(\bm{X}_T)\subseteq R_{\bm{S}}(\bm{Y}_T)$.
\end{enumerate}
Moreover, if $\cA$ is a cone and $\Lambda$ is positively homogeneous, then $R_{\bm{S}}$ also satisfies 
\begin{enumerate}[label=(\roman*)]
\setcounter{enumi}{3}
    \item  \label{item_poshom} \emph{Positive homogeneity:} for all $\bm{X}_T\in L^p_d$ and $c>0 \colon R_{\bm{S}}(c\bm{X}_T) = c R_{\bm{S}}(\bm{X}_T)$.
\end{enumerate}
Finally, if $\cA$ is convex and $\Lambda$ concave, the following properties hold:
\begin{enumerate}[label=(\roman*)]
\setcounter{enumi}{4}
\item  \label{item_convex} \emph{Convexity:} for all $\bm{X}_T, \bm{Y}_T\in L^p_d$ and $\alpha\in[0,1]$ we have 
\begin{align*}
R_{\bm{S}}(\alpha \bm{X}_T +(1-\alpha)\bm{Y}_T)\supseteq \alpha R_{\bm{S}}(\bm{X}_T)+(1-\alpha) R_{\bm{S}}(\bm{Y}_T) \,,    
\end{align*}
\item  \label{item_convexvalues}\emph{$R_{\bm{S}}$ has convex values:} for all $\bm{X}_T \in L^p_d$ the set $R_{\bm{S}}(\bm{X}_T)$ is a convex set in $\RR^d$.
\end{enumerate}
\end{proposition}

The properties in \cref{prop:properties_syst_risk} show that risk measures as in \cref{eq:systRS} are a natural generalisation of univariate monetary risk measures.
Properties \ref{item_S-add}-\ref{item_convex} can be interpreted in analogy to the corresponding properties of scalar monetary risk measures.
Properties \ref{item_uppersets} and \ref{item_convexvalues} are particularly useful for the numerical approximation of the values.
Furthermore, the upper set property allows the efficient communication of risk measurements $R_{\bm{S}}(\bm{X}_T)$ through efficient cash invariant allocation rules (EARs).
For more details see \citep{bib:FRW}.
In addition to these properties, the set $R_{\bm{S}}(\bm{X}_T)$ is closed whenever $\Lambda$ is continuous and $\cA$ is closed, see also \citep[Lemma 2.4 (iii)]{bib:FRW}.
Some of these properties will be discussed in more detail in \cref{sec:intrinsic_systemic_risk}.
\medskip

In analogy to univariate monetary risk measures, the defining management action dictates the institutions in the financial system to adjust their capital holdings by raising capital and investing it in a prespecified asset such that the system can be deemed acceptable.
In the following section, we will explore a different management action without the use of external capital.

\section{Intrinsic measures of systemic risk} \label{sec:intrinsic_systemic_risk}

In this section, we introduce the novel intrinsic measures of systemic risk.
We continue working in the framework of an interconnected network of $d \geq 2$ participants. 
However, in contrast to monetary measures of systemic risk as defined in \cref{def:setval_systemic_risk_measure}, we will not rely on external capital. 
Instead, each participant of the network needs to improve their own position by shifting it towards a specified eligible position. 
A regulatory authority can specify these eligible positions or restrict their choice to certain classes of assets for each participant individually. In this paper, we only discuss the first case, i.e.\ the case where the regulator specifies one eligible position for each participant in the network.
A network in which each participant holds only eligible positions should be acceptable when aggregated, as described in \cref{prop:risk_not_empty}.

In the following section, we define the intrinsic systemic risk measures and derive and discuss their most important properties.

\subsection{Intrinsic measures of systemic risk and their properties}
As in \cref{sec:measures_systemic_risk}, we collect eligible assets and their initial values in a tuple of vectors $\bm{S}=(\bm{s}_0,\bm{S}_T)\in\RR^d_{++}\times (L^p_d)_+$. 
In addition to the random vector $\bm{X}_T$, we also need the initial values of each of the participants' future values, $\bm{x}_0$,
so we extend financial positions to tuples $\bm{X} = (\bm{x}_0, \bm{X}_T) \in \RR^d_{++} \times (L^p_d)_+$.

The shift from the current position towards an eligible position can mathematically be expressed as a convex combination of two random variables. 
Since each participant's position can be altered individually, we extend the notation introduced in \cref{eq:intrinsic_position} element-wise to a multivariate random variable.
For a financial network with a vector of endowments $\bm{X} = (\bm{x}_0, \bm{X}_T) \in \RR^d_{++} \times (L^p_d)_+$, a collection of eligible assets $\bm{S} = (\bm{s}_0, \bm{S}_T) \in \RR^d_{++} \times (L^p_d)_+$, and a vector of coefficients $\bm{\lambda} \in [0,1]^d$ we define the random vector
\begin{align*}
\bm{X}_T^{\bm{\lambda},\bm{S}} &= (1-\bm{\lambda}) \htimes \bm{X}_T  + \bm{\lambda} \htimes \bm{x}_0 \htimes \bm{S}_T \hdiv \bm{s}_0
=\Big( (X^1_T)^{\lambda^1, S^1} \,, \, \ldots \, , \,(X^d_T)^{\lambda^d, S^d}  \Big)^\intercal \in (L^p_d)_+ \,.
\end{align*}
This is the element-wise convex combination of $\XT$ and $\xo \htimes \ST \hdiv \so$ given by the coefficients collected in $\bm{\lambda}$.
In this framework, each participant's financial position $X_T^1, \ldots, X_T^d$ can be altered by a different fraction $\lambda^1, \ldots, \lambda^d$ and by using a different eligible asset $S^1_T, \ldots, S^d_T$, respectively.
Assuming a choice for all $d$ eligible assets has been made, we aim to find all vectors $\bm{\lambda} \in [0,1]^d$ such that the aggregated position $\Lambda(\bm{X}_T^{\bm{\lambda}, \bm{S}})$ belongs to the acceptance set $\cA$.

\medskip

\begin{definition} \label{def:intr_meas_sys_risk}
Let $\Lambda$ be an aggregation function and $\cA$ an acceptance set. 
Let $\bm{S} \in\RR^d_{++}\times (L^p_d)_+$ be a vector-valued eligible asset. 
An \emph{intrinsic measure of systemic risk} is a map  $R_{\bm{S}}^\mathrm{int} \colon \RR^d_{++} \times (L^p_d)_+ \rightarrow \mathcal{P}([0,1]^d)$ defined as
\begin{equation}
R_{\bm{S}}^{\mathrm{int}}(\bm{X}) =  \{ \bm{\lambda} \in [0,1]^d \mid \Lambda( \bm{X}_T^{\bm{\lambda} , \bm{S}} ) \in \cA \} \,.  
\label{eq:set_int}
\end{equation}
\end{definition}

It is important to note that we define intrinsic systemic risk measures on $\RR^d_{++} \times (L^p_d)_+$.
This means $X_T^k$ represents the future value of the asset side of the balance sheet of institution $k$. 
In particular, $\bm{X}_T$ has non-negative values.

\begin{remark}\label{rem:X_positive}
The choice of non-negative $\bm{X}_T$ allows for the most useful operational interpretation of the intrinsic systemic risk measure.
In the context of the element-wise convex combination $\bm{X}_T^{\bm{\lambda},\bm{S}}$, the term 
$(1-\bm{\lambda}) \htimes \bm{X}_T$ should be interpreted as the future value of a system $\bm{X}_T$ after a fraction $\bm{\lambda}$ has been sold.
In general, it would not have the intended operational interpretation if $\bm{X}_T$ denoted the net worth of the institutions, that is, assets minus liabilities, as liabilities would also be scaled.
Only in specific situations, for example with operational costs that are reduced when assets are sold, this interpretation may be accurate.
For this reason, we consider assets and liabilities separately. 
Since in the multivariate framework we make use of aggregation functions, we can incorporate liabilities through them, or if not possible, through the acceptance set.
Network models as studied in \cref{sec:network} go hand in hand with our framework.
This stands in contrast to univariate intrinsic risk measures, where such a restriction is not imposed.
There, if one wants to split assets and liabilities, liabilities need to be incorporated via the acceptance set.
\end{remark}

\medskip

As a technical remark, we note that $R_{\bm{S}}^{\mathrm{int}}$ is always well-defined, since $\emptyset \in \mathcal{P}([0,1]^d)$ is vacuously true.
This is different for the univariate intrinsic risk measure in \cref{intRS}, where $\rho^{\mathrm{int}}_{\cA, S}$ is only well defined if $\frac{x_0}{s_0}S_T\in\cA$, see \cite[below Definition 3.1]{bib:FSint}.
However, an empty risk measure has a similar meaning as the value $+\infty$ for univariate monetary risk measures, namely that the choices of $\cA$, $\Lambda$, and $\bm{S}$ cannot yield an acceptable system. 

Since the aggregation function and the acceptance set can be thought of as restrictions on the system imposed by a regulatory authority, we assume that these objects are given and have certain properties.
In this case, we must choose suitable eligible assets to ensure that $R_{\bm{S}}^{\mathrm{int}}(\bm{X})$ is not an empty set.

\begin{proposition} \label{prop:risk_not_empty}
    Let $\Lambda$ be an aggregation function and let $\cA$ be an acceptance set.
    Let $\bm{X}$ be an unacceptable system in the sense that $\Lambda(\bm{X}_T) \notin \cA$.
    If the eligible asset $\bm{S}$ satisfies $\Lambda(\bm{x}_0 \htimes \bm{S}_T \hdiv \bm{s}_0) \in \cA$, then $\RintX\neq\emptyset$.
\end{proposition}

\begin{proof}
    Notice that $\Lambda(\bm{x}_0 \htimes \bm{S}_T \hdiv \bm{s}_0) = \Lambda(\bm{X}_T^{\bm{1},\bm{S}})$ and hence, $\bm{1} \in \RintX$.
\end{proof}

This requirement makes sense intuitively.
A system in which each agent has fully invested in their eligible asset needs to be acceptable.
This system corresponds to the coefficient vector $\bm{1}$ and marks the end point of the path from $\bm{X}_T$ to $\bm{x}_0 \htimes \bm{S}_T \hdiv \bm{s}_0$.
However, the condition $\Lambda(\bm{x}_0 \htimes \bm{S}_T \hdiv \bm{s}_0) \in \cA$ is not necessary for $\RintX\neq\emptyset$, as can be seen in \cref{fig:intr_rm_different_rhos_S}.

\bigskip

In the following, we collect basic properties of intrinsic measures of systemic risk.
To this end, let $\Lambda$ be an aggregation function and let $\cA$ be an acceptance set.

\begin{proposition}[Monotonicity] \label{claim:systemic_intrinsic_el_monotonicity}
$R_{\bm{S}}^{\mathrm{int}}$ is monotonic in the sense that if $\bm{x}_0\leq \bm{y}_0$ and $\bm{X}_T\leq \bm{Y}_T$ $\PP$-a.s., then $R_{\bm{S}}^{\mathrm{int}}(\bm{X}) \subseteq R_{\bm{S}}^{\mathrm{int}}(\bm{Y})$.
\begin{proof}
Let $\bm{x}_0\leq \bm{y}_0$ and $\bm{X}_T\leq \bm{Y}_T$, and take $\bm{\lambda} \in R_{\bm{S}}^{\mathrm{int}}(\bm{X})$. 
Notice that 
\begin{align*}
(1-\bm{\lambda}) \htimes \bm{X}_T + \bm{\lambda} \htimes \bm{x}_0 \htimes \bm{S}_T \hdiv \bm{s}_0  \leq (1-\bm{\lambda}) \htimes \bm{Y}_T + \bm{\lambda} \htimes \bm{y}_0 \htimes \bm{S}_T \hdiv \bm{s}_0 \,.
\end{align*}
Since $\Lambda$ is non-decreasing, the assertion follows by the monotonicity of $\cA$.
\end{proof}
\end{proposition}

\medskip

\begin{proposition}[Quasi-convexity] \label{claim:systemic_intrinsic_quasi-convex}
If $\cA$ is convex and $\Lambda$ is concave, then $R^{\mathrm{int}}_{\bm{S}}$ is quasi-convex, that is, for all $\alpha \in [0,1]$ and for all $\bm{X},\bm{Y}\in\RR^d_{++}\times (L^p_d)_+$ we have
\begin{align} \label{eq:setvalued_quasi_convex}
R_{\bm{S}}^{\mathrm{int}}(\bm{X}) \cap R_{\bm{S}}^{\mathrm{int}}(\bm{Y}) &\subseteq R_{\bm{S}}^{\mathrm{int}}(\alpha \bm{X} + (1-\alpha)\bm{Y}) \,.
\end{align}
\begin{proof}
If $R_{\bm{S}}^{\mathrm{int}}(\bm{X}) \cap R_{\bm{S}}^{\mathrm{int}}(\bm{Y})=\emptyset$, there is nothing to prove. 
Assume now that $R_{\bm{S}}^{\mathrm{int}}(\bm{X}) \cap R_{\bm{S}}^{\mathrm{int}}(\bm{Y})$ is not empty.
Take any $\bm{\lambda} \in R_{\bm{S}}^{\mathrm{int}}(\bm{X}) \cap R_{\bm{S}}^{\mathrm{int}}(\bm{Y})$, then $\Lambda\big( \bm{X}_T^{\bm{\lambda} , \bm{S}}  \big)$ and $\Lambda\big( \bm{Y}_T^{\bm{\lambda} , \bm{S}}  \big)$ are contained in $\cA$.
Notice that 
\begin{align*}
    (1 - \bm{\lambda})& \htimes (\alpha \bm{X}_T + (1-\alpha) \bm{Y}_T )   + \bm{\lambda} \htimes (\alpha \bm{x}_0 + (1-\alpha)\bm{y}_0 ) \htimes \bm{S}_T \hdiv \bm{s}_0 \\
    &=\alpha \bm{X}_T^{\bm{\lambda} , \bm{S}} + (1-\alpha)\bm{Y}_T^{\bm{\lambda} , \bm{S}} \,.
\end{align*}
By convexity of $\cA$, the convex combination $\alpha \Lambda\big(\bm{X}_T^{\bm{\lambda} , \bm{S}}\big) + (1-\alpha)\Lambda\big(\bm{Y}_T^{\bm{\lambda} , \bm{S}}\big)$ is contained in $\cA$ and by monotonicity of $\cA$ and concavity of $\Lambda$, also $\Lambda\big( \alpha \bm{X}_T^{\bm{\lambda} , \bm{S}} + (1-\alpha)\bm{Y}_T^{\bm{\lambda} , \bm{S}}\big) \in \cA$.
Hence, $\bm{\lambda} \in R_{\bm{S}}^{\mathrm{int}}(\alpha \bm{X} + (1-\alpha)\bm{Y})$.
\end{proof}
\end{proposition}

Property \eqref{eq:setvalued_quasi_convex} is a set-valued version of quasi-convexity. 
The intersection is the set-valued counterpart of a maximum and the subset relation corresponds to the ordering relation $\geq$.
Monotonicity and quasi-convexity are the most important properties a risk measure should satisfy.
Monotonicity implies that all the management actions that make a system of assets $\bm{X}$ acceptable will also make a system $\bm{Y}$ with larger asset values acceptable.
This is also important from a modelling perspective, as the management actions resulting from overestimating risk will not have an adverse effect on the systemic risk. 
Quasi-convexity implements the notion of the diversification principle.
In our set-valued framework, this means that management actions which make both systems $\bm{X}$ and $\bm{Y}$ acceptable will also make any convex combination of them acceptable.

\medskip

We turn to show two regularity properties which are of particular value for the numerical approximation of intrinsic systemic risk measures.

\begin{proposition}[Closed values]
    Let $\Lambda$ be continuous and let $\cA$ be closed.
    Then $\Rint$ has closed values, that is, for all $\bm{X} \in \RR^d_{++} \times (L^p_d)_+$ the set $\RintX$ is a closed subset of $[0,1]^d$. 
\end{proposition}
\begin{proof}
    Let $\RintX \neq \emptyset$ and let $(\bm{\lambda}_n)_{n \in \NN} \subset \RintX$ be a sequence that converges to some $\bm{\lambda} \in [0,1]^d$.
    Notice that $\Xlambdan - \Xlambda = (\bm{\lambda} - \bm{\lambda}_n) \htimes (\XT - \xo \htimes \ST \hdiv \so )$. 
    This gives us 
    \begin{align*}
        \Vert \Xlambdan - \Xlambda \Vert_{L^p_d}
        \leq \max_{k \in \{1,\ldots,d\}} \Big\Vert X_T^k - \frac{x_0^k S_T^k}{s_0^k} \Big\Vert_{L^p}  \cdot \vert \bm{\lambda} - \bm{\lambda}_n \vert_{p} \,.
    \end{align*}
    Since $\Lambda$ is continuous, we get a sequence $\Lambda(\Xlambdan)_{n\in\NN} \subset \cA$ that converges to $\Lambda(\Xlambda)$.
    Since $\cA$ is closed, the limit $\Lambda(\Xlambda)$ is also contained in $\cA$.
    Hence $\bm{\lambda} \in \RintX$.
\end{proof}

Many important examples of acceptance sets are closed, such as the ones associated with Value-at-Risk and Expected Shortfall. 
While in the financial literature closedness is often required to simplify mathematical technicalities, it has also a financial relevance.
Closedness of the acceptance set prevents unacceptable positions to become acceptable through arbitrarily small perturbations, see also \citep[Section 2.2.3]{bib:Munari}.
With continuity of $\Lambda$ this interpretation translates to the set-valued framework.

\medskip

\begin{proposition}[Convex values] \label{claim:riskmeasuresets_are_convex}
If $\Lambda$ is concave and $\cA$ is convex, then for any $\bm{X}\in\RR^d_{++}\times L^p_d$ the set $R_{\bm{S}}^{\mathrm{int}}(\bm{X})$ is convex.
\begin{proof}
Assume that $R_{\bm{S}}^{\mathrm{int}}(\bm{X})$ is not empty, or else there is nothing to show. 
Let $\bm{\lambda}_1, \bm{\lambda}_2 \in R_{\bm{S}}^{\mathrm{int}}(\bm{X})$ and for $\alpha \in [0,1]$ define $\bm{\lambda}_\alpha = \alpha \bm{\lambda}_1 + (1-\alpha) \bm{\lambda}_2 \in [0,1]^d$.
Notice that 
\begin{align*}
\bm{X}_T^{\bm{\lambda}_\alpha,S} &= (1 - \alpha \bm{\lambda}_1 - (1-\alpha) \bm{\lambda}_2) \htimes \bm{X}_T + (\alpha \bm{\lambda}_1 + (1-\alpha) \bm{\lambda}_2 ) \htimes \bm{x}_0 \htimes \bm{S}_T \hdiv \bm{s}_0 \\
&= \alpha \bm{X}_T^{\bm{\lambda}_1,\bm{S}} + (1-\alpha) \bm{X}_T^{\bm{\lambda}_2,\bm{S}} \,. 
\end{align*}

Hence, by concavity of $\Lambda$ we have
\begin{align*}
\Lambda(\bm{X}_T^{\bm{\lambda}_\alpha,\bm{S}}) \geq \alpha \Lambda(\bm{X}_T^{\bm{\lambda}_1,\bm{S}}) + (1-\alpha) \Lambda(\bm{X}_T^{\bm{\lambda}_2,\bm{S}})\,.
\end{align*}

Since both $\Lambda(\bm{X}_T^{\bm{\lambda}_1,\bm{S}})$ and $\Lambda(\bm{X}_T^{\bm{\lambda}_2,\bm{S}})$ are contained in $\cA$, we know by convexity and monotonicity of $\cA$ that $\Lambda(\bm{X}_T^{\bm{\lambda}_\alpha,\bm{S}}) \in \cA$ and thus, $\bm{\lambda}_\alpha \in R_{\bm{S}}^{\mathrm{int}}(\bm{X})$ for all $\alpha \in [0,1]$.
\end{proof}
\end{proposition}

Convexity of the acceptance set corresponds to the diversification principle.
However, since we work with random vectors, we must aggregate them accordingly, which is achieved by requiring that $\Lambda$ is concave.
Indeed, concavity is in line with the diversification principle, as the aggregation of a convex combination of two systems should not reduce the value when compared to the convex combination of two aggregated systems.
This way, convexity of $\cA$ is passed on to $\RintX$.
Furthermore, there is a wide range of numerical methods for the approximation and representation of convex sets. 
In \cref{sec:grid_search_explained}, we will see that convexity allows for a modification of the algorithm for the computation of $\RintX$ that makes the approximation faster for the same accuracy or more accurate when performing the same number of iterations.

\medskip

\begin{remark}
An important and useful property of monetary systemic risk measures is the upper set property stated in \cref{prop:properties_syst_risk} \ref{item_uppersets}.
It motivates the definition of EARs in \citep[Definition 3.3]{bib:FRW}, which can be interpreted as minimal capital requirements for the participants of a financial system, and the communication of which is easier compared to the whole systemic risk measure.
Furthermore, the upper set property lays the basis for the algorithm presented in \citep[Section~4]{bib:FRW}.
In contrast, intrinsic systemic risk measures do not, in general, exhibit this property, as one can see for example in \cref{fig:intr_rm_different_rhos}.
We will attempt to give an intuitive explanation for this.

The upper set property is essentially a consequence of the monotonicity of $\Lambda$ and $\cA$, as shown in \cref{appendix_proofs}.
This means, any participant of an acceptable system is free to add an arbitrary amount of capital through their eligible asset without influencing the acceptability of the system.
In particular, intra- and inter-dependencies of $\bm{X}_T$ and $\bm{S}_T$ are irrelevant.

In the intrinsic approach, however, no external capital is injected into the system. 
Instead, the system is translated element-wise to a system of eligible assets and, in general, the partial order on $L^p_d$ is not enough to compare the resulting positions.
In particular, the set $\RintX$ depends on the interplay of $\bm{X}$ and $\bm{S}$.
For example, assume the elements of $\bm{X}_T = (X_T^1 , X_T^2)^\intercal$ are negatively correlated and the eligible vector $\bm{S}_T \hdiv \bm{s}_0 = (r, r)^\intercal$ is constant. 
If one institution decreases its holding in its original position and increases its holding in the eligible asset, the correlation between the institutions increases.
This in turn can result in an unacceptable aggregate system.
So since the management action in the intrinsic approach does not rely on external capital injections, it is more sensitive to the overall dependency structure of the system and the eligible assets.

The definition of a concept similar to EARs is still possible, as will be discussed in \cref{sec:grid_search_explained}.
However, on a stand-alone basis, institutions are in general not allowed to increase their holdings in the eligible assets beyond the prescribed proportion $\bm{\lambda}$.

For convex $\RintX$ the lack of the upper set property is not a drawback for the computational approximation.
For conic acceptance sets and concave aggregation functions, an additional assumption is sufficient to apply the algorithm described in \cref{sec:grid_search_explained}, see \cref{prop:lines_to_1_contained}.
\end{remark}

\begin{proposition} \label{prop:lines_to_1_contained}
Let $\Lambda$ be concave and let $\cA$ be a cone.
Assume that the eligible asset satisfies $\Lambda(\bm{x}_0 \htimes \bm{S}_T \hdiv \bm{s}_0) \geq 0$.
If $\bm{\lambda} \in R_{\bm{S}}^{\mathrm{int}}(\bm{X})$, then for $\alpha \in [0,1]$ we have $(1-\alpha) \bm{\lambda} + \alpha \bm{1} \in R_{\bm{S}}^{\mathrm{int}}(\bm{X})$.

\begin{proof}
    Let $\alpha \in [0,1]$.
    Since $\bm{\lambda} \in R_{\bm{S}}^{\mathrm{int}}(\bm{X})$, the aggregated position $\Lambda(X^{\bm{\lambda},\bm{S}}_T)$ is contained in $\cA$, and since $\cA$ is a cone, also $(1-\alpha) \Lambda(X^{\bm{\lambda},\bm{S}}_T)$ is contained in $\cA$.
    By concavity of $\Lambda$, we arrive at
    \begin{align*}
        \Lambda \left((1-\alpha) \bm{X}^{\bm{\lambda},\bm{S}}_T + \alpha (\bm{x}_0 \htimes \bm{S}_T \hdiv \bm{s}_0) \right) 
        &\geq (1-\alpha) \Lambda(\bm{X}^{\bm{\lambda},\bm{S}}_T)  + \alpha \Lambda(\bm{x}_0 \htimes \bm{S}_T \hdiv \bm{s}_0) \\
        &\geq (1-\alpha) \Lambda(\bm{X}^{\bm{\lambda},\bm{S}}_T) \in \cA \,,
    \end{align*}
    and thus by monotonicity of $\cA$, the position $\Lambda \left((1-\alpha) \bm{X}^{\bm{\lambda},\bm{S}}_T + \alpha \bm{x}_0 \htimes \bm{S}_T \hdiv \bm{s}_0 \right)$ is acceptable.
    A short calculation shows that 
    \begin{align*}
        (1-\alpha) \bm{X}^{\bm{\lambda},\bm{S}}_T + \alpha (\bm{x}_0 \htimes \bm{S}_T \hdiv \bm{s}_0) = \bm{X}^{(1-\alpha)\bm{\lambda} + \alpha \bm{1}, \bm{S}}_T \,,
    \end{align*}
    proving the assertion.
\end{proof}
\end{proposition}

\medskip

The above result can be interpreted as a set-valued counterpart of 
\begin{align*}
\left\{X^{\lambda, S}_T \mid \lambda\in[\rho^\mathrm{int}_{\cA, S}(X),1]\right\}\subseteq\cA    
\end{align*}
for univariate intrinsic risk measures on conic acceptance sets mentioned in \citep[p.~175]{bib:FSint}. 
However, while in the univariate case this means that any proportion $\lambda\geq \rho^\mathrm{int}_{\cA, S}(X)$ would lead to an acceptable position $X^{\lambda, S}_T$, in the multivariate case it is important that starting from $\bm{\lambda}\in R_{\bm{S}}^{\mathrm{int}}(\bm{X})$ all entries of $\bm{\lambda}$ would have to be increased proportionally to ensure that $\bm{\lambda_2}\geq\bm{\lambda}$ is also included in $R_{\bm{S}}^{\mathrm{int}}(\bm{X})$, see also Figure \ref{fig:grid_search_demo}.

\begin{remark} \label{rem:lines_to_1_contained_convex}
Notice that a statement similar to the one in \cref{prop:lines_to_1_contained} is true if $\Lambda$ is concave, $\cA$ is convex and we demand $\Lambda(\bm{x}_0 \htimes \bm{S}_T \hdiv \bm{s}_0) \in \cA$.
In this case, we know by \cref{claim:riskmeasuresets_are_convex} that $\RintX$ is convex and hence, $(1-\alpha) \bm{\lambda} + \alpha \bm{1} \in R_{\bm{S}}^{\mathrm{int}}(\bm{X})$ for all $\bm{\lambda} \in \RintX$.
\end{remark}

This observation and \cref{prop:lines_to_1_contained} are used in the algorithms described in \cref{sec:grid_search_explained} to approximate values of intrinsic systemic risk measures.

\subsection{Computation of intrinsic measures of systemic risk} \label{sec:grid_search_explained}

As a set-valued measure, the intrinsic systemic risk measure is more difficult to calculate compared to scalar risk measures.
In general, one has to rely on numerical methods to approximate these sets.
A natural approximation would consist of a collection of points $\bm{\lambda} \in \RintX$ which lie close to the boundary $\partial \RintX$.
Assuming a concave aggregation function, we know by \cref{prop:lines_to_1_contained} for conic acceptance sets and \cref{rem:lines_to_1_contained_convex} for convex acceptance sets that all line segments which connect points 
$\bm{\lambda} \in \RintX$ with $\bm{1}$ are contained in $\RintX$.

\medskip

In the following, we will use these results to construct a simple bisection method which approximates the boundary of the intrinsic risk measure with a prespecified accuracy. 
In this section, we assume that either the assumptions in \cref{prop:lines_to_1_contained} or \cref{rem:lines_to_1_contained_convex} hold.
The algorithm is illustrated for a network of two and three participants in \cref{fig:grid_search_demo}. 
Since the intrinsic measure maps into the power set of $[0,1]^d$, we can a priori restrict our search to $[0,1]^d$.

\medskip

We consider the $d$ faces of the cube $[0,1]^d$ which contain the point $\bm{0} \in \RR^d$ and construct a grid on each of these faces. 
The resolution of this grid influences the spacing of the points we approximate on the boundary of $\RintX$.
\begin{enumerate}
    \item Let $\bm{\lambda}_0$ be a point on this grid.
    If $\bm{\lambda}_0 \in \RintX$, add it to the collection of approximation points and proceed with a new grid point $\bm{\lambda}_0$, otherwise continue with Step 2.
    
    \item Do a bisection search along the line connecting $\bm{\lambda}_0$ and $\bm{1}$. 
    To this end, define $\bm{\lambda}_{a_0} = \bm{\lambda}_0 \notin \RintX$ and $\bm{\lambda}_{b_0} = \bm{1} \in \RintX$ as the initial end points.
    
    \item Iterate over $k \geq 1$ and construct a sequence that tends to a point $\bm{\lambda}$ on the boundary of $\RintX$.
    In each iteration, define $\bm{\lambda}_k = \frac{1}{2}(\bm{\lambda}_{a_{k-1}} + \bm{\lambda}_{b_{k-1}})$ and check whether the aggregated position corresponding to $\bm{\lambda}_k$ is acceptable.
    
    \item If $\Lambda(\bm{X}_T^{\bm{\lambda}_k , \bm{S}}) \in \cA$, set $\bm{\lambda}_{a_k} = \bm{\lambda}_{a_{k-1}}$ and $\bm{\lambda}_{b_k} = \bm{\lambda}_k$, otherwise $\bm{\lambda}_{a_k} = \bm{\lambda}_k$ and $\bm{\lambda}_{b_k} = \bm{\lambda}_{b_{k-1}}$.
    
    \item Stop this procedure at $n \in \NN$ at which the distance $\Vert \bm{\lambda}_{b_{n}} - \bm{\lambda}_{a_{n}} \Vert$ is smaller than some desired threshold $\epsilon>0$ and take $\hat{\bm{\lambda}} = \bm{\lambda}_{b_n}$ as the approximation of $\bm{\lambda}$. 
    This also covers the rare case that for some $k$, $\bm{\lambda}_k = \bm{\lambda}$. Add $\hat{\bm{\lambda}}$ to the collection of approximations and repeat the procedure with a new grid point $\bm{\lambda}_0$.
\end{enumerate}

By definition, $\hat{\bm{\lambda}} = \bm{\lambda}_{b_n}$ lies in $\RintX$ and is arbitrarily close to the boundary, as $\Vert \hat{\bm{\lambda}} - \bm{\lambda} \Vert \leq \Vert \bm{\lambda}_{b_n} - \bm{\lambda}_{a_n} \Vert = 2^{-n} \Vert 1-\bm{\lambda}_0 \Vert \leq 2^{-n}\sqrt{d}$. 
Since for any point $\bm{\lambda}_0$ on the grid we have $1-\bm{\lambda}_0 \geq 0$, we also know that $\bm{\lambda}_{a_n} \leq \bm{\lambda} \leq \bm{\lambda}_{b_n}$.

We repeat this procedure for all points $\bm{\lambda}_0$ in the grid on the faces of the cube containing $\bm{0}$.
This establishes an approximation of the boundary of $\RintX$ as a collection of points $\{ \hat{\bm{\lambda}}_k \}_{k=1}^N \subset \RintX$, where $N$ is the number of grid points.
Notice that this algorithm approximates only the boundary marked in green in \cref{fig:grid_search_demo}.
The rest of the boundary is approximated by all lines connecting the algorithmically found points on the faces of $[0,1]^d$ with $\bm{1}$.

\begin{figure}[ht]
    \centering
    \includegraphics[scale=0.65]{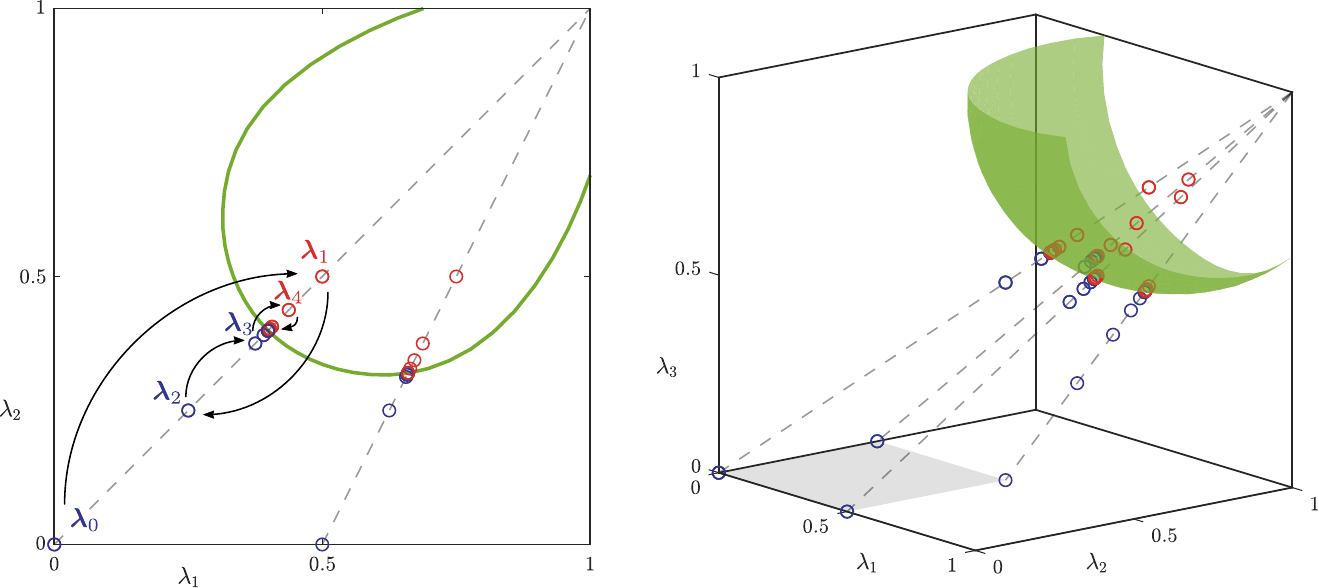}
    \caption{Illustration of grid search algorithm on $[0,1]^2$ and $[0,1]^3$.}
    \label{fig:grid_search_demo}
\end{figure}

\begin{remark}
This collection of points constitutes an inner approximation of the set $\RintX$.
Notice that an outer approximation of the set as part of a $\bm{v}$-approximation as described in \citep[Definition 4.1]{bib:FRW} is in general not possible, since the values of $\Rint$ are in general not upper sets and therefore $\RintX \not\subset \RintX - \bm{v}$, for $\bm{v} \in \RR^d_{++}$.
However, this is not a drawback, since we are only interested in elements of $\RintX$.
\end{remark}

\medskip

We can easily communicate this approximation of the set $\RintX$ for small dimensions.
In particular, since all line segments connecting points of the approximation and $\bm{1}$ are contained in $\RintX$, we get a `discrete cover' of $\RintX$ by lines.
For a fine grid, this may be precise enough for practical purposes.
However, interpolation between these lines is, in general, not possible if the acceptance set is not convex.

If, however, the acceptance set is convex, then we know by \cref{claim:riskmeasuresets_are_convex} that $R_{\bm{S}}^{\mathrm{int}}(\bm{X})$ is convex.
In this case, we can define a stronger approximation as the convex hull of $\{\hat{\bm{\lambda}}_k \}_{k=1}^N$ and the vector $\bm{1}$,
\begin{align*}
    \RintXhat \coloneqq \conv \{ \{ \hat{\bm{\lambda}}_k \}_{k=1}^N \cup \{\bm{1} \}\} \subset \RintX \,.
\end{align*}
In particular, any interpolation between the line segments is also contained in $\RintX$. 
This allows for the tradeoff between accuracy at the boundary of $\RintX$ and substantially faster computation by coarsening the grid on the faces, while still covering the majority of $\RintX$.

\begin{remark}
This algorithm can be applied to high-dimensional systems at the expense of considerably longer runtime and memory usage.
\citet[Remark 4.3]{bib:FRW} suggest to reduce the dimension of the problem by dividing the set of institutions into groups with equal capital requirements for the computation of their measure of systemic risk. 
This is also possible in the intrinsic framework.
In analogy to \citep[Example 2.1 (iv)]{bib:FRW}, for $k < d$ groups we can restrict the risk measurement to vectors of the form
$\bm{\lambda} = (\lambda_1, \ldots, \lambda_1, \lambda_2,\ldots,\lambda_2,\ldots,\lambda_k, \ldots,\lambda_k)^\intercal \in [0,1]^d$.
However, since the values of intrinsic systemic risk measures are not upper sets, players cannot, in general, deviate from this position.
This means a system is not guaranteed to remain acceptable if an institution in group $j \in \{1,\ldots,k\}$ increases its position in the eligible asset, that is, chooses to sell a greater fraction of its position than $\lambda_j$.
\end{remark}

\medskip

\begin{remark}
In analogy to \citep[Definition 3.3]{bib:FRW}, we can define a notion similar to EARs for intrinsic systemic risk measures.
For a convex, closed risk measurement $\RintX \notin \{\emptyset , [0,1]^d \}$ we define the set of minimal points as
\begin{align*}
    \mathrm{Min}\,\RintX = \{ \bm{\lambda} \in [0,1]^d \mid (\bm{\lambda} - [0,1]^d) \cap \RintX = \{ \bm{\lambda} \}  \} \,.
\end{align*}
However, since $\RintX$ is not an upper set in general, it is not true that for $\bm{\lambda}^* \in \mathrm{Min}\,\RintX$ the set $(\bm{\lambda}^* + [0,1]^d) \cap [0,1]^d$ is a maximal subset of $\RintX$.
In fact, it is not necessarily a subset at all.
In particular, there are points $\bm{\lambda}^* \in \mathrm{Min}\,\RintX$ such that $\bm{\lambda}^* + \epsilon \bm{e}_k \notin \RintX$ for any $\epsilon > 0$ and some standard unit vector $\bm{e}_k \in \RR^d$.
This means agents cannot deviate from allocations in $\mathrm{Min}\,\RintX$.
To tackle this problem and allow small perturbations without loosing acceptability, we could further restrict the set of minimal points for $\epsilon > 0$ to
\begin{align*}
    \mathrm{Min}_\epsilon\,\RintX = \{ \bm{\lambda} \in \mathrm{Min}\,\RintX \mid \bm{\lambda} + \epsilon [0,1]^d \in \RintX \} \,.
\end{align*}
However, a priori this set is not guaranteed to be non-empty.
Take for example a pointed, closed, convex cone, $C\subsetneq\RR_{++}^d$.
Then for any $\bm{x}\in (0,1)^d$ the set $(C + \bm{x}) \cap [0,1]^d$ has only one minimal point, $\bm{x}$, and $\bm{x} + \epsilon \bm{e}_k \notin (C + \bm{x}) \cap [0,1]^d$, for any $\epsilon >0$ and standard unit vector $\bm{e}_k$.
However, it remains to be investigated whether convex intrinsic systemic risk measurements can take this form.
\end{remark}

\bigskip

From a practical perspective, the objective of the regulator might be to make a network acceptable with as little alteration to existing positions as possible.
What exactly this means depends on the choice of the objective function.
The most straightforward choice would be to minimise the overall percentage change of all positions in the network.
In that case, we are interested in all $\bm{\lambda} \in \RintX$ with minimal sum over their components, or equivalently on $[0,1]^d$, with minimal $1$-norm,
\begin{align} \label{eq:argmin_min_sum}
    \argmin_{\bm{\lambda} \in \RintX}  \bm{\lambda}^\intercal \bm{1}  =  \argmin_{\bm{\lambda} \in \RintX}  \vert \bm{\lambda} \vert_1 \,.
\end{align}
Alternatively, one could minimise the total nominal change in the sense of the value $\bm{x}_0 \htimes \bm{\lambda}$. 
The optimisation for this, and in fact any weighted sum with non-negative weights $\bm{w}$, can be written in the form of \cref{eq:argmin_min_sum} by minimising over the set $\bm{w} \htimes \RintX = \{ \bm{w} \htimes \bm{\lambda} \mid \bm{\lambda} \in \RintX  \}$.
So in the following, we will concentrate on the case in \eqref{eq:argmin_min_sum}.

To simplify the problem, we define for $k\geq0$ the plane $E_k = \{ \ell \in [0,1]^d \mid \ell^\intercal \bm{1} = k \}$, on which each element has the same $1$-norm.
For increasing $k \in [0,d]$ these planes contain elements with increasing $1$-norms.
In particular, the first nonempty intersection $E_k \cap \RintX$ for increasing $k$ contains all points with minimal $1$-norm in $\RintX$,
\begin{align*}
\argmin_{\bm{\lambda} \in \RintX}  \bm{\lambda}^\intercal \bm{1} = 
E_{k_{\min}}\cap\RintX \,,
\end{align*}
where $k_{\min}=\min\left\{k \in [0,d] \mid E_k\cap \RintX\neq\emptyset\right\}$.
So if we are interested in these minimal points, we do not need to approximate the whole boundary of the risk measurement. 
Instead, we can take advantage of this observation and adapt the algorithm to find only the minimal points.
This also reduces the computational load.

\medskip

To implement this procedure, we define the orthogonal complement of $\bm{1} \in \RR^d$, $\bm{1}^{\perp} = \{ \ell \in \RR^d \mid \ell^\intercal \bm{1} = 0 \}$, and we write $E_k = (\frac{k}{d}\bm{1} + \bm{1}^{\perp}) \cap [0,1]^d$.
The idea is to do a modified bisection search, where in each iteration we check whether the intersection of $E_k$ and $\RintX$ is empty or not.
The algorithm is described below and its modification using the method described in \cref{rem:reduce_planes} is illustrated in \cref{fig:grid_search_with_plane}.
\begin{enumerate}
    \item Generate a grid on the plane $\bm{1}^{\perp}$, such that it covers the whole cube $[0,1]^d$ when translated along the vector $\bm{1}$.
    
    \item Define $k_{a_0} = 0$ and $k_{b_0} = d$ as the initial end points of the search.

    \item In each iteration, define $k_\ell = \frac{1}{2}( k_{a_{\ell-1}} + k_{b_{\ell-1}} )$ and translate the grid from $\bm{1}^\perp$ to $E_{k_\ell}$. 
    Calculate which grid points lie in $E_{k_\ell} \cap \RintX$.
    If the intersection contains no grid points, set $k_{a_{\ell}} = k_\ell$ and $k_{b_{\ell}} = k_{b_{\ell-1}}$, otherwise set $k_{a_{\ell}} = k_{a_{\ell-1}}$ and $k_{b_{\ell}} = k_{\ell}$.
    
    \item Repeat Step 3 until $k_{b_\ell} - k_{a_\ell} < \delta$, for some prespecified threshold $\delta > 0$.
\end{enumerate}

This leaves us with a sequence of planes that converge to $E_{k_{\min}}$ and we can easily check which grid points lie in $E_{k_{\min}}\cap\RintX$.

\medskip

\begin{remark} \label{rem:reduce_planes}
    If the acceptance set is convex, we can reduce the computational time even further with the help of \cref{lemma:convex_set_and_planes}.
    This method is briefly outlined below.
    In addition to the steps described above, we also keep track of the current acceptable grid points in $E_{k_\ell} \cap \RintX$. 
    Then, whenever we consider grid points of a non-empty intersection $E_{k_\ell - \epsilon} \cap \RintX$ for some $\epsilon > 0$, we compare them to the current acceptable grid points.
    If all the grid points in $E_{k_\ell - \epsilon} \cap \RintX$ are contained in the set of grid points in $(E_{k_\ell} \cap \RintX) - \frac{\epsilon}{d} \bm{1}$, we can, by \cref{lemma:convex_set_and_planes} , restrict our search to only the former grid points.
    
    \medskip
    
    Alternatively, we can increase the accuracy of the algorithm by refining the grid.
    Whenever we use \cref{lemma:convex_set_and_planes} and restrict our search from a grid on some $E_k$ to a smaller grid on $E_{k-\epsilon}$, we can decrease the step size between grid points such that the number of grid points per iteration stays the same.
\end{remark}

\begin{figure}
    \centering
    \includegraphics[scale=0.62]{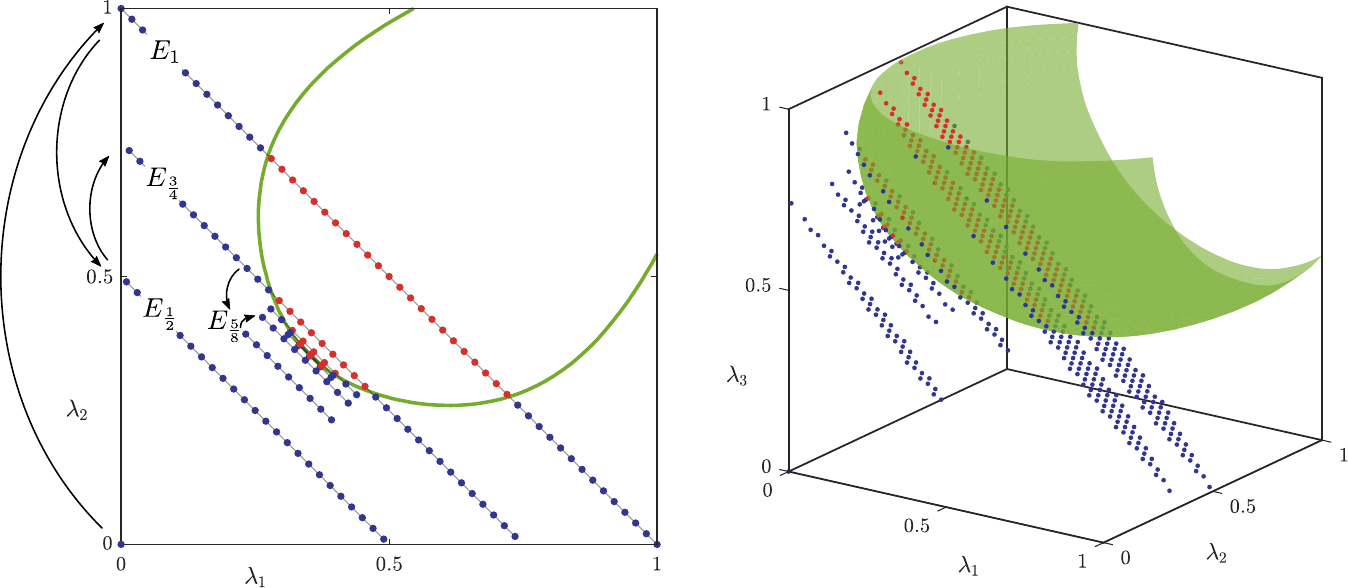}
    \caption{Visualisation of grid search algorithm to find minimal points.}
    \label{fig:grid_search_with_plane}
\end{figure}

\begin{lemma} \label{lemma:convex_set_and_planes}
Let $A \subset \RR^d$ be a closed, convex set.
Let $A_k = \{ \bm{x}\in A \mid \bm{x}^\intercal \bm{1} = k \}$.
If there exists $\epsilon > 0$ such that $A_{k-\epsilon} \subset A_k - \frac{\epsilon}{d}\bm{1}$, then for all $\delta > \epsilon$ also $A_{k-\delta} \subset A_{k-\epsilon} - \frac{\delta-\epsilon}{d} \bm{1}$.
\end{lemma}

The proof of \cref{lemma:convex_set_and_planes} is given in \cref{appendix_proofs}.


\section{Dual representations} \label{sec:dual_representation}

It is a classical result, which can be found for example in \citep[Section 4.2]{bib:FS}, that a convex proper lower semi-continuous risk measure $\rho_\cA: L^\infty \to \RR$ with $\bm{S}=(1,\bm{1})$ can be written in the form 
\begin{align*}
\rho_{\cA}(X)=\sup\limits_{\QQ\in\cM(\PP)}\left\{\EE^{\QQ}\left[-X\right]-\alpha(\QQ)\right\}
\end{align*}
where $\cM(\PP)$ denotes the set of all probability measures that are absolutely continuous with respect to $\PP$ and $\alpha$ denotes the minimal penalty function defined by 
\begin{equation}\label{eq:penalty}
    \alpha(\QQ)=\sup\limits_{X\in\cA}\EE^{\QQ}\left[-X\right].
\end{equation}
This can be seen as considering all possible probabilistic models, with their plausibility and closeness to $\PP$ being conveyed in the penalty function $\alpha$. 
The value of the risk measure then corresponds to the worst case expectation over all possible models, penalised by $\alpha$. 

\citet[Section 3.4]{bib:FSint} derive a similar result for scalar intrinsic risk measures, whereas \citet{bib:AR} provide the dual representations for monetary systemic risk measures with constant eligible assets 
under appropriate assumptions on the underlying acceptance set and aggregation function.

\medskip

In this section, we derive the dual representation for the intrinsic systemic risk measure. 
We denote by $\rho_\cA$ the scalar risk measure associated with the acceptance set $\cA$ and a constant eligible assets $S=(1, 1) \in \RR_+ \times L^\infty_+$.
In the following, we consider only the subspace $(L^\infty_d)_+$ and we assume that the aggregation function $\Lambda$ is concave.
Moreover, we assume that $\cA$ is convex and weak$^*$ closed.
Note that if $\cA$ is convex, we know by \cref{prop:accset_rm_correspondence} that $\rho_{\cA}$ is convex.
Furthermore, weak$^*$ closedness of $\cA$ is equivalent to $\rho_{\cA}$ being weak$^*$ lower semi-continuous as well as to $\rho_{\cA}$ satisfying the Fatou property, see for example \citep[Theorem 4.31]{bib:FS}.

Let $g: \RR^d \to \RR$ be the Legendre-Fenchel conjugate of the convex function $f: \RR^d \to \RR$ defined by $f(\bm{x})=-\Lambda(-\bm{x})$, that is, 
\begin{align*}
    g(\bm{z})=\sup\limits_{\bm{x}\in\RR^d}\left(\Lambda(\bm{x})-\bm{z}^\intercal \bm{x}\right) \,.
\end{align*}
Let $\cM_d(\PP)$ be the set of all vector probability measures $\QQ=(\QQ_1,\ldots,\QQ_d)^\intercal$ whose components $\QQ_k$ are in $\cM(\PP)$, $k \in \{1,\ldots,d\}$, and let $\alpha$ be the penalty function defined in \cref{eq:penalty}. 
We recall the definition of a systemic penalty function introduced in \citep[Definition 3.1]{bib:AR}.

\begin{definition}\label{def:systemicPenalty}
The function $\alpha^\mathrm{sys}:\cM_d(\PP)\times(\RR^d_+\setminus\left\{0\right\})\rightarrow\RR\cup\left\{+\infty\right\}$ defined by 
\begin{align*}
    \alpha^\mathrm{sys}(\QQ,w)=\inf\limits_{\substack{\bS\in\cM(\PP)\colon\\ \forall k \colon w_k\QQ_k\ll\bS}}\left\{\EE^{\bS}\left[g\left(\bm{w}\htimes\frac{\dd\QQ}{\dd\bS}\right)\right]+\alpha(\bS)\right\}
\end{align*}
for $\QQ\in\cM_d(\PP),\: \bm{w}\in\RR^d_+\setminus\left\{0\right\}$ is called the systemic penalty function. 
\end{definition}

\medskip

Now we can formulate the dual representation of intrinsic systemic risk measures.
\begin{proposition}\label{claim:dualRepresentations}
Let $\bm{S}\in\RR^d_{++}\times(L^\infty_d)_+$ be an eligible asset, let $\Lambda$ be a concave aggregation function and let $\cA$ be a weak$^*$ closed, convex acceptance set. 
Assume that $\rho_\cA(0)\in\Lambda(\RR^d)$. 
Then the intrinsic measure of systemic risk $R_{\bm{S}}^\mathrm{int} \colon \RR^d_{++} \times (L^\infty_d)_+ \rightarrow \mathcal{P}([0,1]^d)$ defined in \cref{eq:set_int} admits the following dual representation,
\begin{align*}
    R_{\bm{S}}^\mathrm{int}(\bm{X})=\bigcap\limits_{\substack{\QQ\in\cM_d(\PP),\\ \bm{w}\in \RR_+^d\setminus \left\{0\right\}}}\Bigg\{\bm{\lambda} \in [0,1]^d \; \bigg| \; \begin{split}
        \bm{\lambda}^\intercal \big( \bm{w}\htimes\EE^{\QQ}\left[\bm{x}_0\htimes \bm{S}_T\hdiv \bm{s}_0-\bm{X}_T\right] \big) \geq \ldots \\ \ldots \bm{w}^\intercal\EE^{\QQ}\left[-\bm{X}_T\right]-\alpha^\mathrm{sys}(\QQ,\bm{w}) \end{split}
    \Bigg\} \,.
\end{align*}
\end{proposition}
\begin{proof}

From \citep[Proposition 3.4]{bib:AR}, we have
\begin{equation}
\rho_\cA(\Lambda(\bm{X_T}))=\sup\limits_{\substack{\QQ\in\cM_d(\PP),\\ \bm{w}\in\RR^d_+\setminus\left\{0\right\}}}\left(\bm{w}^{\intercal}\EE^{\QQ}\left[-\bm{X}_T\right]-\alpha^\mathrm{sys}(\QQ,\bm{w})\right) \,.
\label{eq:dual_rho_Lambda}
\end{equation} 

By \cref{prop:accset_rm_correspondence}, we can write
\begin{align*}
R_{\bm{S}}^\mathrm{int}(\bm{X}) = \left\{\bm{\lambda}\in\left[0,1\right] \mid \Lambda \big(\Xlambda \big)\in\cA\right\} 
=\left\{ \bm{\lambda}\in\left[0,1\right] \mid \rho_{\cA}\big(\Lambda\big(\Xlambda\big)\big) \leq0 \right\} \,.
\end{align*}
Together with \cref{eq:dual_rho_Lambda} it follows that $\bm{\lambda} \in \left[0,1\right]^d$ lies in $R_{\bm{S}}^{\mathrm{int}}(\bm{X})$ if and only if 
\begin{align*}
\forall\QQ\in\cM_d(\PP),\bm{w}\in\RR^d_+ \setminus \{0\}: 
\bm{w}^{\intercal}\EE^{\QQ}\left[-\Xlambda\right]-\alpha^\mathrm{sys}(\QQ,\bm{w})\leq 0
\end{align*}
or, rewritten,
\begin{align*}
&\forall\QQ\in\cM_d(\PP),\bm{w}\in\RR^d_+ \setminus \{0\} \colon \\
&\bm{\lambda}^{\intercal} \big(\bm{w}\htimes\EE^{\QQ}\left[\bm{X}_T-\bm{x}_0\htimes \bm{S}_T\hdiv \bm{s}_0\right] \big) \leq \alpha^\mathrm{sys}(\QQ,\bm{w}) + \bm{w}^{\intercal}\EE^{\QQ}\left[\bm{X}_T\right] \,.
\end{align*}
From here, the claim follows.
\end{proof}

\medskip

The dual representation of $R_{\bm{S}}^{\rm{int}}$ given in \cref{claim:dualRepresentations} can be interpreted in a similar way as in \citet[p.~147f]{bib:AR}. 
Consider a network consisting of $d$ institutions, represented by the elements of $\bm{x}_0$ and $\bm{X}_T$, as well as society, see \cref{sec:network} for an example of such a network model.
The dual representation collects the possible restructuring actions in the presence of model uncertainty and weight ambiguity. 

Society is assigned a probability measure $\bS$ and each institution is assigned its own probability measure $\QQ_k$ along with a weight $w_k$ with respect to society. 
The penalty function $\alpha^\mathrm{sys}$ combines two penalty terms.
One is the distance of the network to society, captured as the multivariate g-divergence of $\QQ$ with respect to $\bS$, $\EE^\bS\left[g(\bm{w}\htimes\frac{\dd\QQ}{\dd\bS})\right]$. 
The other is the penalty $\alpha(\bS)$ incurred for the choice of $\bS$.
The penalty function $\alpha^\mathrm{sys}$ is then given as the infimum of the sum of these penalties over all choices of $\bS$.

Finally, a vector of fractions $\bm{\lambda} \in [0,1]^d$ is deemed feasible for a specific choice of $\QQ$ and $\bm{w}$, if the weighted sum of expected return of the eligible asset held in the restructured portfolios of institutions exceeds the weighted expected negative return of the original positions $\bf{X}$ held in the restructured portfolios of the institutions penalised by $\alpha^\mathrm{sys}({\QQ,\bm{w}})$, 
\begin{align*}
\bm{w}^\intercal\EE^\QQ\left[\bm{\lambda} \htimes \bm{x}_0\htimes \bm{S}_T\hdiv \bm{s}_0\right]\geq \bm{w}^\intercal \EE^\QQ\left[-(1-\bm{\lambda})\htimes \bm{X}_T\right]-\alpha^\mathrm{sys}(\QQ,\bm{w}) \,.
\end{align*}
To be considered as a feasible action to compensate the systemic risk in the system, the vector of fractions $\bm{\lambda}$ has to be deemed feasible for all possible choices of probability measures $\QQ$ and weights $\bm{w}$.


\section{The network approach - a simulation study} \label{sec:network}

In this section, we will investigate the effects of the management actions underlying intrinsic systemic risk measures on networks as originally proposed by \citet{EisenbergNoe2001}.
We will complement their model with an additional sink node called \emph{society} and define the aggregation function as the net equity of society after receiving the clearing payments as described in \citep[Section 4.4]{bib:AR}, see also \citep[Section 5.2]{bib:FRW}.
This enables us to derive statements about the repercussions of an under-capitalised financial system and to monitor the impact of intrinsic management actions on the wider economy.
Furthermore, this study provides insights into current regulatory policies and challenges the current approach to capital regulation.

\subsection{Network model}

In the following, we recall the network model.
An illustration of the network structure can be seen in \cref{fig:financial_network}.
A financial system consists of $d+1$, $d \geq 2$, nodes.
Nodes $\{1,\ldots,d\}$ represent financial institutions participating in the network and node $0$ represents society.
The network is interconnected via liabilities towards each other. 
Throughout this section, we assume for simplicity that liabilities are deterministic, whereas future endowments of participants of the network are represented by a random vector $\bm{X}_T$ with an initial value $\bm{x}_0$. 
For $i,j \in \{0, \ldots, d\}$ let $L_{ij} \geq 0$ denote the nominal liability of node $i$ towards node $j$. 
Self-liabilities are disregarded, that is $L_{ii} = 0$ for all $i\in \{0,\ldots,d\}$.
Node $0$ is a sink node, which means, we assume that it has no liabilities towards the other $d$ nodes in the system, that is $L_{0i} = 0$ for $i \in \{1,\ldots,d\}$. 
In line with \citep{bib:FRW}, we interpret society as part of the wider economy.
That means node $0$ represents all outside factors which are not explicitly part of the financial system. 
This allows for the assumption $\bm{X}_T\geq0$, implicitly assuming that any operational costs and other factors that could make the endowments negative are treated as liabilities towards society.
In particular, we assume that all institutions have liabilities towards society, that is for $i \in \{1,\ldots,d\}$ we have $L_{i0} > 0$.
Furthermore, we define relative liabilities for $i,j\in\{0,\ldots,d\}$, $i \neq 0$ as 
\begin{align*}
    \Pi_{ij}= \frac{L_{ij}}{\hat{L}_i} \; \text{ with } \; \hat{L}_i = \sum_{j=0}^d L_{ij} >0 \,,
\end{align*}
where $\hat{L}_i$ is the aggregate nominal liability of $i$ towards all other nodes in the network.

\medskip

At time point $T$, liabilities are cleared.
This means that all participants in the network repay all or, if not possible, part of their liabilities. 
For a realised state $\bm{x}_T \coloneqq \bm{X}_T(\omega) \in \RR^d_+$ at time $T$ we collect these payments in a vector $p(\bm{x}_T) = (p_1(\bm{x}_T) , \ldots , p_d(\bm{x}_T))^\intercal \in \RR^d_+$, so that node $i$ pays node $j$ the amount $\Pi_{ij} p_i(\bm{x}_T)$.
We call $p(\bm{x}_T)$ a \emph{clearing payment vector} if it solves the fixed point problem
\begin{align*} 
p_i(\bm{x}_T) = \min \bigg\{\hat{L}_i \; , \; x_T^i + \sum_{j=1}^d \Pi_{ji} p_j(\bm{x}_T) \bigg\} \,, \quad i \in \{1, \ldots, d\}\,.
\end{align*}

In the case that institution $i$ stays in business and no default occurs, it pays all of its liabilities to the rest of the network, $p_i(\bm{x}_T) = \hat{L}_i$.
In the case of default, the payment is equal to the realised wealth, $x_T^i$, plus the income from the other participants of the network, $\sum_{j=1}^d \Pi_{ji} p_j(\bm{x}_T)$.
The clearing payment vector can be calculated with the `fictitious default algorithm' introduced by \citet[p.~243, see also Lemma 3 and Lemma 4]{EisenbergNoe2001}, or by means of an appropriate optimization problem with linear constraints.

\medskip

We choose the aggregation function $\Lambda:\RR^d \to \RR$ for the intrinsic systemic risk measure to represent the impact of the financial system on society.
For $\beta \in (0,1)$ we define
\begin{align} \label{eq:aggregation_fct_network_model}
\Lambda(\bm{x})=\sum\limits_{i=1}^d\Pi_{i0}p_i(\bm{x}) - \beta \sum_{i=1}^d L_{i0} \,.
\end{align} 
While other choices of an aggregation function are possible, this is an aggregation function often used in the network setting, see e.g.\ \citet{bib:FRW}, \citet{bib:AR}.
This aggregation function quantifies the (weighted) difference between what the society obtains from the other nodes in the network at the time of clearing and between what it was promised, i.e.\ the sum of the liabilities of all network participants towards society. 
Note that the first term in \cref{eq:aggregation_fct_network_model} lies in the interval $[0 , \sum_{i=1}^d L_{i0}]$.
If we were to use only this term, then every system would be acceptable with regard to $\cA_{\VaR_\alpha} = \{ X_T \in L^p \mid \VaR_\alpha(X_T) \leq 0 \}$ or $\cA_{\ES_\alpha} = \{ X_T \in L^p \mid \ES_\alpha(X_T) \leq 0 \}$.
Therefore, we subtract the second term $\beta \sum_{i=1}^d L_{i0}$.
The interpretation is in line with \citep{bib:FRW}.
We deem a system acceptable if the Value-at-Risk, respectively the Expected Shortfall, of the aggregated payments to society does not exceed the negative of a percentage $\beta$ of the total liabilities to society.
For the Value-at-Risk, this means that a system is acceptable if with at least the probability $(1-\alpha)$ at least the percentage $\beta$ of the liabilities towards society can be repaid.

\begin{figure}[ht]
	\begin{center}
		\begin{tikzpicture}[>=stealth',shorten >=1pt,auto,node distance=2cm]
			\node[state](q3) 		  							{$3$};
			\node[]		(qHilfe4)	[right of=q3] 				{};
			\node[state](q4)      	[right of = qHilfe4]		{$0$};
			\node[] 	(qHilfe12)	[left of =q3]  				{};
			\node[state](q1) 		[above left of=qHilfe12] 	{$1$};
			\node[state](q2) 		[below left of=qHilfe12] 	{$2$};
			
			\path[->]          (q1)  edge  [bend left = 20]  	node {$L_{12}$} (q2);
			\path[->]          (q2)  edge  [bend left =20]		node {$L_{21}$} (q1);
			\path[->]          (q3)  edge                 		node {$L_{32}$} (q2);
			\path[->]          (q1)  edge                 		node {$L_{13}$} (q3);		
			\path[->]      	   (q1)  edge  [bend left = 20]		node {$L_{10}$} (q4);
			\path[->]      	   (q2)  edge  [bend right = 20]	node {$L_{20}$} (q4);
			\path[->]      	   (q3)  edge     					node {$L_{30}$} (q4);
		\end{tikzpicture}
		\caption{Network of financial institutions including society as node $0$.} 
		\label{fig:financial_network}
	\end{center}
\end{figure}
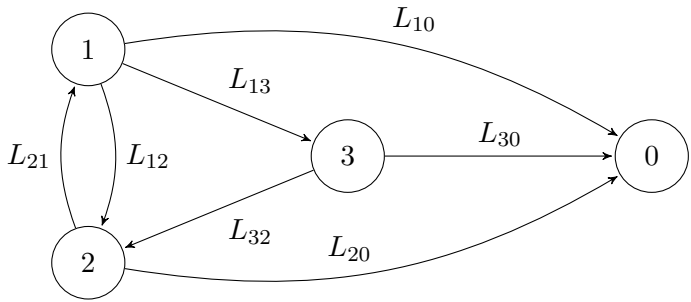

\begin{remark} \label{rem:entensions_eisenberg_noe}
It is possible to extend this network model in  a variety of ways.
A natural extension would incorporate random liabilities.
Furthermore, one can incorporate the illiquidity of assets during fire sales, see for example \citep{bib:liquidityRisk, bib:feinstein_firesales, bib:Contagion!}, bankruptcy costs as in \citep{RogersVeraart2013}, cross holdings as in \citep{bib:elsinger}, or a combination of the above as in \citep{bib:weber_weske}.
\end{remark}

\subsection{Numerical case studies}

In order to get a better understanding of the behaviour of intrinsic systemic risk measures under different model parameters, it is helpful to first consider a simple network of $d=2$ institutions.
The following model parameters define the base model.
Throughout this section, we will adjust them separately to see their effects.
The marginal distributions of the agents' wealths at time $T$ are assumed to be beta distributions, $X_T^k \sim \mathrm{Beta}(a_k,b_k)$, with $a_k = 2, b_k = 5$, $k \in \{1,2\}$.
We choose the initial value $\bm{x}_0$ such that the expected return of each institution is $15\%$, that is, $\EE[\bm{X}_T] = 1.15\, \bm{x}_0$.
The eligible assets have a log-normal marginal distribution, $S_T^k \sim \log\cN(\mu_k, \sigma_k^2)$, $k \in \{1,2\}$.
We specify $\mu_k$ and $\sigma_k$ such that the expectations of the eligible assets are equal to the expectations of the agent's positions, $\EE[\bm{X}_T] = \EE[\bm{S}_T]$, and such that the variances are a fifth of the variances of $\bm{X}_T$, $\VV[\bm{S}_T] = 0.2 \, \VV[\bm{X}_T]$.
We set the initial cost of the eligible assets such that the expected return is $10\%$, $\EE[\bm{S}_T] = 1.1\, \bm{s}_0$. 
In this model, the eligible assets are more secure in the sense that their distributions are a lot narrower.
But, in turn, this security comes at an additional cost, since $\bm{x}_0 \leq \bm{s}_0$.

Dependence is incorporated via a Gaussian copula with a correlation $\rho \in [-1,1]$ between $X_T^1$ and $X_T^2$.
The eligible assets are uncorrelated with each other and with $X_T^k$, $k\in\{1,2\}$.

Furthermore, we assume that the agents have symmetric liabilities to each other, $L_{12} = L_{21} = 0.6$, and to society, $L_{10} = L_{20} = 0.2$.
We use the aggregation function specified in \cref{eq:aggregation_fct_network_model}.
We deem a system acceptable if the Expected Shortfall at probability level $\alpha = 5\%$ of the aggregated network outcome 
is less or equal $0$, that is, $\bm{\lambda} \in \RintX$ if and only if $\ES_\alpha(\Lambda(\Xlambda))\leq 0$.

The following numerical studies consist of $10^5$ simulated samples of the described multivariate distribution.
The step size of the grid on the axes is set to $0.05$ and the size of the interval at which we stop the bisection search is set to $10^{-6}$. 
The computations are done using MATLAB.

In the following, for $\bm{\lambda} \in \RintX$ we will refer to $\Xlambda$ as an intrinsic system and to $\Lambda(\Xlambda)$ as an aggregate intrinsic system.
We use analogous nomenclature for the monetary case.

In the following figures, we depict the boundaries of $\RintX$ on the left-hand side and the boundaries of $\RmonX$ on the right-hand side.

\paragraph{Influence of the dependency structure}
In \cref{fig:intr_rm_different_rhos}, we illustrate the risk measures for different correlations between the elements of $\bm{X}_T$.
We observe that as the correlation between the elements of $\bm{X}_T$ increases, the risk measures become smaller in the sense that $R^{\mathrm{int}}_{\bm{S}}(\bm{X}_\rho) \subset R^{\mathrm{int}}_{\bm{S}}(\bm{X}_{\hat{\rho}})$ for correlations $\rho > \hat{\rho}$.
This is expected, since higher correlation between the participants of the network results in higher probability of cascades of defaults and hence, the inability to repay society.
Furthermore, allocations which are acceptable for highly correlated agents are also acceptable if the correlation decreases while other dependencies stay unaltered.

We also observe that the lines representing the boundaries of all the sets meet in two points on the boundary of $[0,1]^2$.
This comes from the fact that $\bm{X}$ and $\bm{S}$ are uncorrelated, so if one agent translates fully to the eligible asset, the correlation between $X_T^1$ and $X_T^2$ becomes irrelevant.
A similar statement can be made for the monetary risk measures, where the whole system is deemed acceptable when enough capital is added to either of the two agents. 

In this symmetric case, we observe the intuitive result that the cheapest way to acceptance, in the sense that $\lambda_1 + \lambda_2$ or $k_1 + k_2$ is minimal, is when both agents adjust their position equally, that is $\lambda_1 = \lambda_2$ or $k_1=k_2$.

\begin{figure}[h]
    \centering
    \includegraphics[scale=1.1]{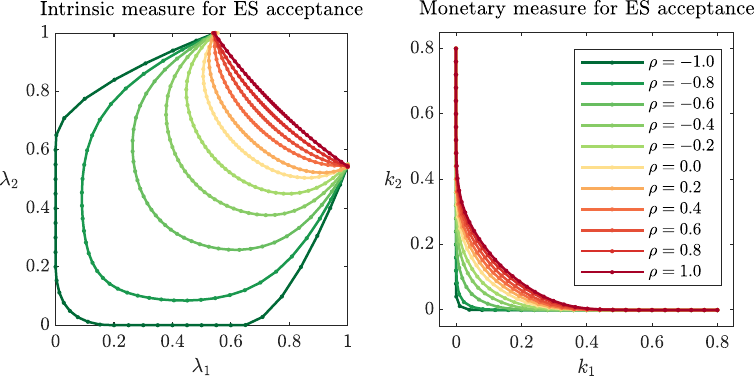}
    \caption{Visualisation of the influence of different correlations between the agents' positions.}
    \label{fig:intr_rm_different_rhos}
\end{figure}

\medskip

In the next scenario, we adjust the previous setup and assume different correlations between institution $1$ and their eligible asset.
In \cref{fig:intr_rm_different_rhos_bw_X1_S1}, we sample distributions such that $X_T^1$ and $S_T^1$ are correlated with parameter $\rho$, while the rest of the system remains uncorrelated.
The yellow lines are the same as the yellow lines in \cref{fig:intr_rm_different_rhos}, as in both cases, $X^1_T, X^2_T, S^1_T$, and $S^2_T$ are uncorrelated.
We observe an accumulation point on the set $\{\lambda_1 = 1\}$, since $X_T^2$, $S_T^2$, and $S_T^1$ are uncorrelated, and no accumulation point on $\{\lambda_2 = 1\}$, since $X_T^1$ and $S_T^1$ are correlated.
Furthermore, close to $\{\lambda_1 = 1\}$, the sets are almost identical to the ones depicted in \cref{fig:intr_rm_different_rhos}.
This means that when institution $1$ has almost fully invested in their eligible asset, the correlation between $X_T^1$ and $S_T^1$ has little effect on the management actions of institution $2$.
However, close to $\{\lambda_2 = 1\}$, we observe that negative correlation results in less strict management actions for institution $1$, whereas positive correlation results in stricter management actions compared to \cref{fig:intr_rm_different_rhos}.

\begin{figure}[h!]
    \centering
    \includegraphics[scale=1.1]{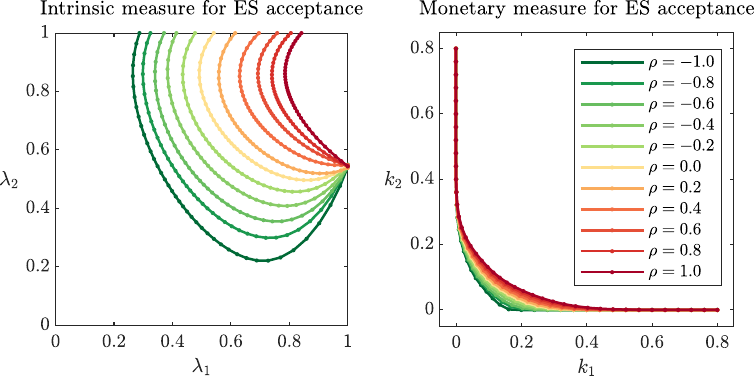}
    \caption{Visualisation of the influence of different correlations between the position of one agent and their corresponding eligible asset.}
    \label{fig:intr_rm_different_rhos_bw_X1_S1}
\end{figure}

\medskip

Now we move on to discuss the effect of correlated eligible assets.
For this we adapt the base setup by setting the correlation between $S_T^1$ and $S_T^2$ to $\rho$ and all other correlations to $0$.
The resulting risk measures are depicted in \cref{fig:intr_rm_different_rhos_S}.
Compared to the monetary measure of systemic risk, the intrinsic measure is more sensitive to the choice of eligible assets.
In this example, the aggregate position of the system consisting of only the eligible assets is acceptable up to approximately a correlation of $0.24$. 
Correlations higher than this result in an intrinsic risk measure which does not include $\bm{1}$.
However, this does not mean that the set $\RintX$ is necessarily empty.
For completeness we have included the risk measurement for $\rho=0.4$, which demonstrates that the condition $\Lambda(\bm{x}_0 \htimes \bm{S}_T \hdiv \bm{s}_0) \in \cA$ is conservative in the sense that it is sufficient but not necessary for a non-empty risk measurement.
In particular, it can be possible to construct an acceptable system even if $\Lambda(\bm{x}_0 \htimes \bm{S}_T \hdiv \bm{s}_0) \notin \cA$.
However, this boundary cannot be calculated with the algorithm described in \cref{sec:grid_search_explained} and we check the grid points on the full grid on $[0,1]^2$ instead.

Whereas from a regulatory point of view, a first thought might be to adapt the agents' positions with only one single eligible asset or asset class, this result underlines the importance of a diversified network.
In particular, it is beneficial to choose eligible assets which are negatively correlated or uncorrelated with each other. 
This effect is less apparent for monetary measures, since adding eligible assets inherently increases the overall capital level and the dependence between eligible assets plays a secondary role.
Nevertheless, adding positively correlated eligible assets increases the correlations between the agents.
\begin{figure}[h!]
    \centering
    \includegraphics[scale=1.1]{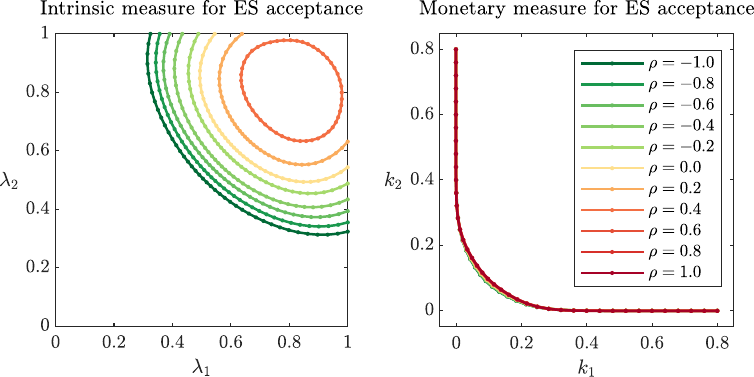}
    \caption{Visualisation of the influence of different correlations between the eligible assets.}
    \label{fig:intr_rm_different_rhos_S}
\end{figure}

\paragraph{Influence of the liability structure}
In this paragraph, we investigate the impact of the liability structure on the systemic intrinsic risk measure.
We use the parameters of the base case with uncorrelated $X_T^k,S_T^k$, $k \in \{1,2\}$.

\medskip

In \cref{fig:intr_rm_different_liabilities_to_society}, we leave the bilateral liabilities of the institutions at $L_{12} = L_{21} = 0.6$ and we gradually increase both their liabilities towards society from $L_{10} = L_{20} = 0.1$ to $L_{10} = L_{20} = 0.2$.
We observe that the risk measurements are decreasing with increasing liabilities towards society.
This is expected, as the term $\beta \sum_{i=1}^d L_{i0}$ in \cref{eq:aggregation_fct_network_model} has a linear influence on the aggregation.
It is noticeable that for low liabilities the monetary systemic risk measurements have considerable parts intersecting with $\RR^2 \setminus \RR^2_{+}$.
This means that if one institution raises enough capital, then the other one can extract capital from the system while the system remains acceptable.
Compared to \cref{fig:intr_rm_different_rhos}, where changes in correlation changed the shape of the sets and made them `pointier', changes in liabilities to society rather translate the whole set.
\begin{figure}[h!]
    \centering
    \includegraphics[scale=1.1]{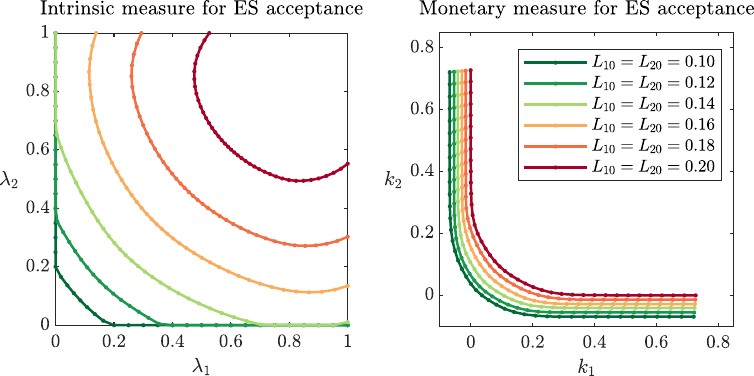}
    \caption{Visualisation of the influence of increasing symmetric liabilities towards society.}
    \label{fig:intr_rm_different_liabilities_to_society}
\end{figure}

\medskip

In the next example, we keep liabilities towards society constant at $L_{10} = L_{20} = 0.2$ and vary bilateral liabilities between the agents.
At first, the result in \cref{fig:intr_rm_different_symm_bilateral_liabilities} might seem counter-intuitive, as both intrinsic and monetary risk measurements increase with increasing liabilities.
However, increasing liabilities between the institutions in this network essentially means adding capital to the system. 
In particular, if one institution is doing poorly and goes bankrupt while the other is doing well, it will still receive the full payment from the other institution.
The higher this payment, the higher is the payment from the defaulting institution towards society.
It is an interesting observation that the intrinsic risk measurements appear to converge to a `maximal set'.
This set is very close to the one represented by the red line in \cref{fig:intr_rm_different_symm_bilateral_liabilities}.
In the case of the monetary measure, it is not clear from this preliminary investigation whether the sets approach a half-space which is supported at a point with $k_1 = k_2$.
See also \cref{fig:liabilities_limit} in \cref{appendix_liabilities}.
However, the higher the bilateral liabilities and the more external capital one of the agents holds, the more capital the other agent can extract from their position.
This could be a dangerous feature of the monetary approach.
\begin{figure}[h!]
    \centering
    \includegraphics[scale=1.1]{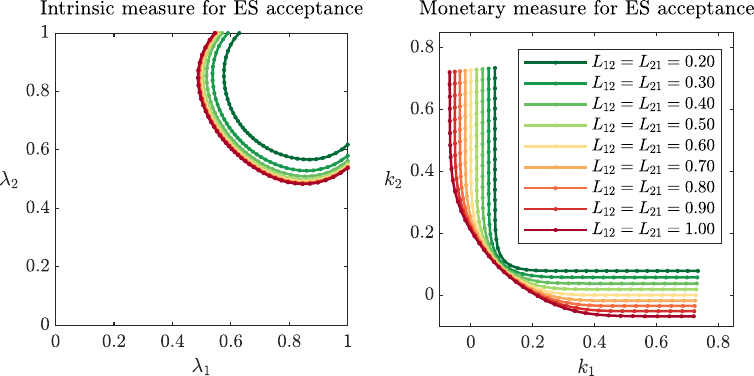}
    \caption{Visualisation of the influence of increasing symmetric bilateral liabilities between the agents.}
    \label{fig:intr_rm_different_symm_bilateral_liabilities}
\end{figure}

\paragraph{Influence of volatility}
We briefly discuss how the variance of an agent's position influences the risk measurements.
In \cref{fig:intr_rm_different_distribution_1player}, we start with the base case with uncorrelated random variables.
We then decrease the variance of the beta distribution of agent $1$ (from green to red) while keeping the expectation at $\frac{a_1}{a_1+b_1} = \frac{2}{2+5}$.
As expected, we observe that both risk measurements increase with decreasing variance.
In particular, agent $1$ needs to adjust their position less in comparison with agent $2$.
\begin{figure}[h!]
    \centering
    \includegraphics[scale=1.1]{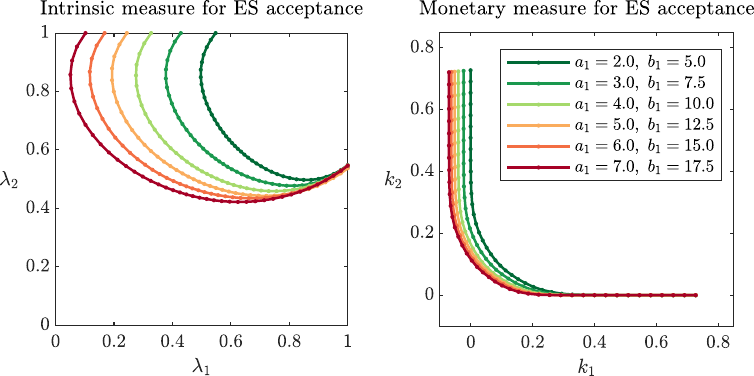}
    \caption{Visualisation of the influence of increasing variance of one agent.}
    \label{fig:intr_rm_different_distribution_1player}
\end{figure}

\paragraph{Intrinsic management actions and aggregate network outcomes}
We conclude this section with a discussion about the aggregated outcomes resulting from management actions of $\bm{\lambda}$ and $\bm{z}$ on the boundaries of $\RintX$ and $\RmonX$, respectively. 
In the following, we assume that all $X_T^k, S_T^k$, $k \in \{1,\ldots,d\}$ are uncorrelated.
Furthermore, since there are more players, we adjust the liability structure.
For $d=4$ we set $L_{ij} = 0.6$ and $L_{i0} = 0.23$ and for $d=20$ we set $L_{ij} = 0.2$ and $L_{i0} = 0.25$, for $i,j \in \{1,\ldots,d\}$.
The rest of the parameters remain unchanged.

\medskip

In \cref{fig:agg_pos_4_and_20_players}, the histograms of the aggregated outcomes of systems with four and 20 agents are depicted.
The aggregate eligible systems $\Lambda(\bm{x}_0 \htimes \bm{S}_T \hdiv \bm{s}_0)$ are presented in green, the intrinsic systems $\Lambda(\bm{X}_T^{\bm{\lambda}, \bm{S}})$ in yellow, the monetary systems
$\Lambda(\bm{X}_T + \bm{k})$ in blue, and the original unacceptable systems
$\Lambda(\bm{X}_T)$ in red.
The vectors $\bm{\lambda}$ and $\bm{k}$ lie in $\RintX$ and $\RmonX$, respectively, and they are multiples of $\bm{1}$, that is, $\bm{\lambda} = \lambda \bm{1}$ and $\bm{k} = k \bm{1}$.
In particular, the Expected Shortfall of $\Lambda(\bm{x}_0 \htimes \bm{S}_T \hdiv \bm{s}_0)$ is negative and the Expected Shortfall of $\Lambda(\bm{X}_T^{\bm{\lambda}, \bm{S}})$ and $\Lambda(\bm{X}_T + \bm{k})$ is approximately equal to $0$.
The dots of corresponding colour indicate the minimum of the support of the histogram. Note that we chose the points of the form $\bm{\lambda}=\lambda\bm{1}$ and $\bm{k}=k\bm{1}$ because we work with perfectly symmetric systems which means that these particular choices are also the solution of the scalarised minimisation problem \eqref{eq:argmin_min_sum} and its corresponding version for the monetary risk measures. If the system were not symmetric, other choices of $\bm{\lambda}$ and $\bm{k}$ might be of interest. 
One reason for this is that the solution to \eqref{eq:argmin_min_sum} does not take the form of $\bm{\lambda}=\lambda\bm{1}$. 
Additionally, in systems where there is a clear distinction between large and small firms, a weighted sum minimisation problem can be formulated by assigning different capital prices to each group of firms.

First we notice that the minimum value of the aggregated intrinsic position (yellow dot) is greater than the one of the monetary position (blue dot).
Since the expected shortfall of both positions is approximately $0$, the mass in the tail of the distribution of the aggregate monetary position is more spread out.
In this sense, the worst cases of the intrinsic positions are milder compared to the monetary positions.
Furthermore, we observe that the distribution of the intrinsic system is more right-skewed.
This means that the intrinsic system is more likely to repay more of its liabilities to society.
From a regulatory perspective this is a valuable insight, as it demonstrates that changing the structure of a financial system can be more beneficial to society than elevating it by external capital.

\begin{figure}[h!]
    \centering
    \includegraphics[scale=1.1]{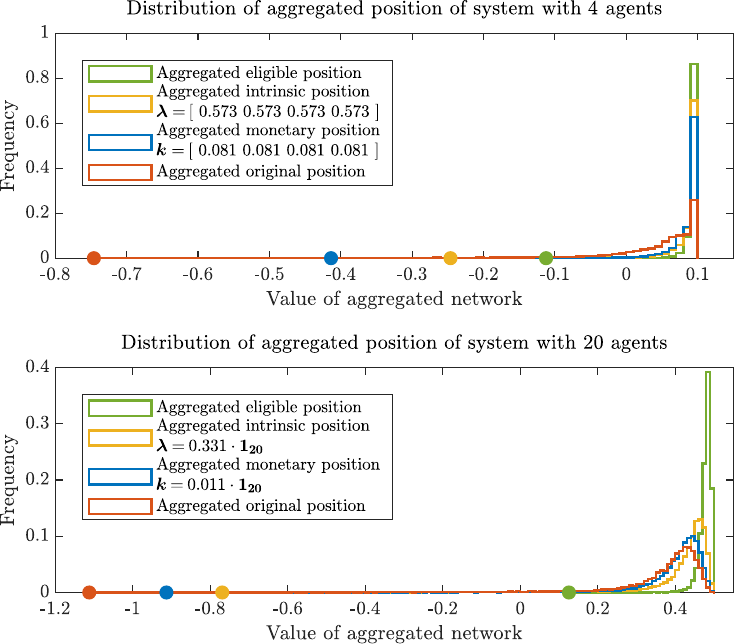}
    \caption{Visualisation of the influence of increasing liabilities of one agents to the other.}
    \label{fig:agg_pos_4_and_20_players}
\end{figure}

\medskip

From a preliminary statistical analysis, we observe that the variance of the aggregate intrinsic system is in general slightly smaller compared to the variance of the aggregate monetary system, whereas the expected value is slightly bigger. 
However, this needs to be verified in a more elaborate study.

In order to understand the effect of additionally incurred costs due to the implementation of the management actions, we have also considered a crude implementation of transaction costs and cost of debt.
For the intrinsic measure we implemented transaction costs of $50$ basis points of $\bm{\lambda}\htimes \bm{x}_0$ once for selling the original position and once for buying the eligible asset.
In the monetary case, we implemented transaction costs for buying the eligible asset and the cost of debt of $2.64\%$ for raising the necessary external capital $\bm{k}$.
In both cases, the resulting aggregate positions do not change considerably and we omit these results here.

\section{Conclusion and possible extensions} \label{sec:conclusion}

We proposed a novel approach to measure systemic risk.
We challenged the paradigm of using external capital injections to the financial system and suggested realistic management actions that fundamentally change the structure of the system such that it can become less volatile and less correlated.
We developed two algorithms, one to approximate the boundary of intrinsic systemic risk measurements and the other to find specific minimal points without calculating the whole boundary.
Furthermore, we derived a dual representation of intrinsic systemic risk measures.
Finally, on the basis of numerical case studies, we demonstrated that intrinsic systemic risk measures are a useful tool to analyse and mitigate systemic risk.

\medskip

We mention here possible extensions and further research avenues.

The notion of EARs associated with monetary systemic risk measures needs to be adapted to intrinsic systemic risk measures in a meaningful way.
In particular, the absence of the upper set property calls for the introduction of further properties to allow agents in the network to deviate from the suggested risk measurement in a controlled way without loosing acceptability.
Furthermore, it remains to be shown that this would allow to group together similar institutions to reduce the complexity of the model and computational time.

The Eisenberg-Noe model has seen numerous extensions which can also be applied to our framework.
These include in particular extensions to random liabilities and the incorporation of illiquidity during fire sales.
This is interesting, as intrinsic measures rely on selling parts of risky portfolios and buying safer assets.
While the primary objective of intrinsic systemic risk measures is to mitigate risk to prevent crises, it is valuable to know what restrictions it faces during a crisis.

Moreover, it is necessary to study more asymmetrically structured networks and develop `fair' allocation rules, as for example according to the contribution of liabilities to society.

Furthermore, it would be valuable to conduct an empirical case study including systemically relevant banks and their actual liabilities towards each other and towards participants of the greater economy.

\appendix
\section{Proofs}
\label{appendix_proofs}

In the following, we prove the assertions made in \cref{prop:properties_syst_risk}.
\begin{proof}[Proof of \cref{prop:properties_syst_risk}] 
Let $\XT\in L^p_d$ and $\bm{\ell}\in\RR^d$.
$\bm{S}$-additivity follows directly from
\begin{align*}
\Rmon(\XT+\ell\htimes\ST\hdiv\so) 
&=\left\{\bm{k}\in\RR^d\mid\Lambda(\XT+(\bm{\ell}+\bm{k})\htimes\ST\hdiv\so)\in\cA\right\}\\
&=\left\{\hat{\bm{k}}\in\RR^d\mid\Lambda(\XT+\hat{\bm{k}}\htimes\ST\hdiv\so)\in\cA\right\}-\bm{\ell}=\RmonX-\bm{\ell}.
\end{align*}

Monotonicity follows from monotonicity of $\cA$ and $\Lambda$. 
For $\XT,\bm{Y}_T\in L^p_d$ with $\XT\leq \bm{Y}_T$ $\PP$-a.s.~and $\bm{k}\in\RmonX$ we have $\Lambda(\bm{Y}_T+\bm{k}\htimes\ST\hdiv\so) \geq \Lambda(\XT+\bm{k}\htimes\ST\hdiv\so) \in \cA$.
This implies $\bm{k}\in\Rmon(\bm{Y}_T)$.

$\bm{S}$-additivity and monotonicity together imply that the values of $\Rmon$ are upper sets. 
Let $\XT\in L^p_d$ and $\bm{y}\in\RR^d_+$. Then $\XT-\bm{y}\htimes\ST\hdiv\so\leq\XT$ and we have 
\begin{align*}
\Rmon(\XT)+\bm{y}=\Rmon(\XT-\bm{y}\htimes\ST\hdiv\so)\subseteq\Rmon(\XT) \,.
\end{align*}
Since the above holds for any $\bm{y}\in\RR^d_+$, the claim follows.

For positive homogeneity, assume that $\cA$ is a cone, $\Lambda$ is positively homogeneous and let $X\in L^p_d$ and $c>0$. 
Notice that in this case $\Lambda(c\XT+\bm{k}\htimes\ST\hdiv\so)\in\cA$ is equivalent to $\Lambda(\XT+\frac{1}{c}\bm{k}\htimes\ST\hdiv\so) \in\cA$.
Therefore,
\begin{align*}
\Rmon(c\XT) 
&=\left\{\bm{k}\in\RR^d\mid \Lambda(\XT+\frac{1}{c}\bm{k}\htimes\ST\hdiv\so) \in\cA \right\}\\
&=c\left\{\hat{\bm{k}}\in\RR^d\mid \Lambda(\XT+\hat{\bm{k}}\htimes\ST\hdiv\so) \in\cA \right\}=c\RmonX \,.
\end{align*}

Finally, for properties \ref{item_convex} and \ref{item_convexvalues} we assume that $\cA$ is convex and $\Lambda$ is concave and let $\XT, \bm{Y}_T\in L^p_d$, $\alpha\in[0,1]$. 
To show convexity, let $\bm{x}\in\RmonX$, $\bm{y}\in\Rmon(\bm{Y}_T)$. 
We get 
\begin{align*}
\Lambda(\alpha\XT&+(1-\alpha)\bm{Y}_T+(\alpha\bm{x}+(1-\alpha)\bm{y})\htimes\ST\hdiv\so )\\
&=\Lambda\big(\alpha(\XT+\bm{x}\htimes\ST\hdiv\so)+(1-\alpha)(\bm{Y}_T+\bm{y}\htimes\ST\hdiv\so)\big)\\
&\geq\alpha\Lambda(\XT+\bm{x}\htimes\ST\hdiv\so)+(1-\alpha)\Lambda(\bm{Y}_T+\bm{y}\htimes\ST\hdiv\so)\in\cA \,,
\end{align*}
where the element inclusion is implied by the convexity of $\cA$. By monotonicity of $\cA$ the assertion follows.

To show that $\RmonX$ has convex values, let $\bm{k},\bm{\ell}\in\RmonX$. 
Notice that 
\begin{align*}
    \XT&+(\alpha\bm{k}+(1-\alpha)\bm{\ell})\htimes\ST\hdiv\so = \alpha(\XT+\bm{k}\htimes\ST\hdiv\so)+(1-\alpha)(\XT+\bm{\ell}\htimes\ST\hdiv\so) \,. 
\end{align*}
The assertion follows as in the proof of \ref{item_convex}.
\end{proof}

\bigskip

In the following, we prove \cref{lemma:convex_set_and_planes}.
\begin{proof}[Proof of \cref{lemma:convex_set_and_planes}]
    Assume by contradiction that for some $\epsilon>0$ with $A_{k-\epsilon} \subset A_k - \frac{\epsilon}{d}\bm{1}$ there exists an $\delta>\epsilon$ and an $\bm{x}_\delta \in A_{k-\delta}$ such that $\bm{x}_\delta \notin A_{k-\epsilon} - \frac{\delta - \epsilon}{d}\bm{1}$.
    Since $A$ is closed, $\hat{\bm{x}} = \argmin_{\bm{x} \in A_{k-\epsilon}} \Vert \bm{x} - \bm{x}_\delta \Vert$ exists and is contained in $A_{k-\epsilon}$.
    Notice that by assumption, $\hat{\bm{x}} \neq \bm{x}_\delta + \frac{\delta-\epsilon}{d}\bm{1}$.
    Therefore there exists $\bm{y} \in \RR^d \setminus \{\bm{0}\}$ with $\bm{y}^\intercal \bm{1} = 0$ such that $\bm{x}_\delta = \hat{\bm{x}} - \frac{\delta-\epsilon}{d}\bm{1} + \bm{y}$.
    Furthermore, notice that $\Vert \bm{x}_\delta - (\hat{\bm{x}}+\beta \bm{y}) \Vert = \Vert (1-\beta) \bm{y} - \frac{\delta-\epsilon}{d}\bm{1} \Vert$ is decreasing in $\beta \in [0,1]$.
    In particular for $\beta \in (0,1]$, $\hat{\bm{x}} + \beta\bm{y}$ cannot lie in $A_{k-\epsilon}$, and since $(\hat{\bm{x}} + \beta\bm{y})^\intercal \bm{1} = k-\epsilon$, $\hat{\bm{x}} + \beta\bm{y} \notin A$.
    
    Now let $\bm{x}_\epsilon = \frac{\epsilon}{\delta} \bm{x}_\delta + (1- \frac{\epsilon}{\delta}) (\hat{\bm{x}} + \frac{\epsilon}{d}\bm{1}) = \hat{\bm{x}} + \frac{\epsilon}{\delta} \bm{y}$.
    From the previous observation, we see that $\bm{x}_\epsilon \notin A$.
    However by assumption, $\hat{\bm{x}} + \frac{\epsilon}{d}\bm{1} \in A_k$.
    So by convexity of $A$, $\bm{x}_\epsilon$ must lie in $A$.
    This is a contradiction and therefore, such an $\bm{x}_\delta$ cannot exist.
\end{proof}

\section{Note on increasing bilateral liabilities} \label{appendix_liabilities}

This appendix complements the discussion around \cref{fig:intr_rm_different_symm_bilateral_liabilities}.
As bilateral liabilities between two agents increase, both systemic risk measures increase. 
The following figure illustrates the `limit set' of the intrinsic systemic risk measure.
For monetary systemic risk measurements it is not clear from our simulations whether they converge to a half-space or not.

\begin{figure}[h!]
    \centering
    \includegraphics[scale=1.1]{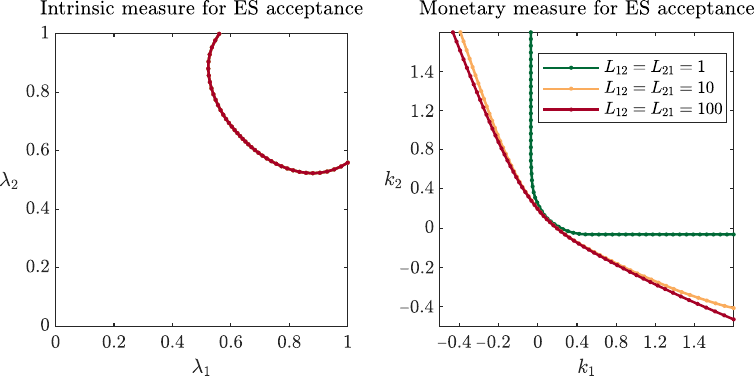}
    \caption{Visualisation of the limit set for increasing bilateral liabilities.}
    \label{fig:liabilities_limit}
\end{figure}

\medskip

\paragraph{Acknowledgements}
The authors are thankful to Gabriela Kováčová for her contributions to the proof of \cref{lemma:convex_set_and_planes}. Jana Hlavinov{\'a} and Birgit Rudloff acknowledges support from the OeNB anniversary fund, project number 17793, and from the Vienna Graduate School on Computational Optimization, Austrian Science Fund (FWF): W1260-N35.

\bibliographystyle{mybibstyle}
\bibliography{bibliography}

\end{document}